\documentclass[journal,twoside]{IEEEtran}
\usepackage{amsfonts}
\usepackage{graphicx}
\usepackage{color}
\usepackage{amsmath,amsfonts,amssymb,amsthm,epsfig,epstopdf,url,array}
\usepackage{url,textcomp}
\usepackage{authblk}
\usepackage{cite}
\usepackage{mathtools}

\newtheorem{theorem}{Theorem}
\newtheorem{lemma}{Lemma}

\begin{document}
\title{Achievable Diversity Order of HARQ-Aided Downlink NOMA Systems}
\author{Zheng~Shi,
        Chenmeng~Zhang,
        Yaru~Fu,
        Hong~Wang,
        Guanghua~Yang,
        and Shaodan~Ma
\thanks{Copyright (c) 2015 IEEE. Personal use of this material is permitted. However, permission to use this material for any other purposes must be obtained from the IEEE by sending a request to pubs-permissions@ieee.org.}
\thanks{Manuscript received April 23, 2019; revised September 16, 2019; accepted October 10, 2019. This work was supported in part by the National Key Research and Development Program of China under Grant 2017YFE0120600, in part by National Natural Science Foundation of China under Grants 61801192, 61801246 and 61601524, in part by Guangdong Basic and Applied Basic Research Foundation under Grant 2019A1515012136, in part by the Science and Technology Planning Project of Guangdong Province under Grant 2018B010114002, in part by the Science and Technology Development Fund, Macau SAR (File no. 0036/2019/A1 and File no. SKL-IOTSC2018-2020), in part by Open Research Foundation of National Mobile Communications Research Laboratory of Southeast University under Grant 2018D09, and in part by the Research Committee of University of Macau under Grant MYRG2018-00156-FST. The associate editor coordinating the review of this paper and approving it for publication was Daniel Benevides da Costa. (Corresponding author: Guanghua Yang.)}
\thanks{Zheng Shi and Guanghua Yang are with the Institute of Physical Internet and the School of Intelligent Systems Science and Engineering, Jinan University, Zhuhai 519070, China (e-mails:shizheng0124@gmail.com, ghyang@jnu.edu.cn).}
\thanks{Chenmeng~Zhang is with the School of Telecommunications Engineering, Xidian University, Xi'an 710000, China (e-mail:zcm\_stu@163.com).}
\thanks{Yaru Fu is with the Department of Information Systems Technology and Design, Singapore University of Technology and Design, Singapore (e-mail:yaru\_fu@sutd.edu.sg).}
\thanks{H. Wang is with the School of Communication and Information Engineering, Nanjing University of Posts and Telecommunications, Nanjing 210003, China, and also with the National Mobile Communications Research Laboratory, Southeast University, Nanjing 210096, China (e-mail: wanghong@njupt.edu.cn).}
\thanks{Shaodan Ma is with the State Key Laboratory of Internet of Things for Smart City and the Department of Electrical and Computer Engineering, University of Macau, Macao, China (e-mail: shaodanma@um.edu.mo).}
}
\maketitle
\begin{abstract}
The combination between non-orthogonal multiple access (NOMA) and hybrid automatic repeat request (HARQ) is capable of realizing ultra-reliability, high throughput and many concurrent connections particularly for emerging communication systems. This paper focuses on characterizing the asymptotic scaling law of the outage probability of HARQ-aided NOMA systems with respect to the transmit power, i.e., diversity order. The analysis of diversity order is carried out for three basic types of HARQ-aided downlink NOMA systems, including Type I HARQ, HARQ with chase combining (HARQ-CC) and HARQ with incremental redundancy (HARQ-IR). The diversity orders of three HARQ-aided downlink NOMA systems are derived in closed-form, where an integration domain partition trick is developed to obtain the bounds of the outage probability specially for HARQ-CC and HARQ-IR-aided NOMA systems. The analytical results show that the diversity order is a decreasing step function of transmission rate, and full time diversity can only be achieved under a sufficiently low transmission rate. It is also revealed that HARQ-IR-aided NOMA systems have the largest diversity order, followed by HARQ-CC-aided and then Type I HARQ-aided NOMA systems. Additionally, the users' diversity orders follow a descending order according to their respective average channel gains. Furthermore, we expand discussions on the cases of power-efficient transmissions and imperfect channel state information (CSI). Monte Carlo simulations finally confirm our analysis.

\end{abstract}
\IEEEpeerreviewmaketitle
\hyphenation{HARQ}
\section{Introduction}\label{sec:int}
\subsection{Related Works}
Non-orthogonal multiple access (NOMA) has been recently identified as a promising candidate for future cellular multiple access \cite{ding2017application,fu2019mode,wang2018precoding,do2018improving,zhao2019joint}. With the paradigms of Internet of Things (IoT), ultra-dense heterogeneous networks (HetNets), etc., many potential NOMA related technologies have been developed for various new scenarios \cite{an2017achieving,gao2019auction,yang2018non}. Unlike the traditional orthogonal multiple access (OMA) techniques, the NOMA technique capitalizes on the superposition coding to accommodate multiple users simultaneously with different power levels yet at a sharing physical resource (frequency/time/code) \cite{dai2015non,islam2016power,Yaru_TVT,masaracchia2019pso,zhao2019beamforming}. Furthermore, the successive interference cancellation (SIC) is performed at the decoding stage to differentiate the desired signal from interfering signals. NOMA not only exploits the multi-user diversity to increase the spectral efficiency, but also seeks to strike a balance between the throughput and user fairness. The notable advantages of NOMA inspire us to integrate it with other potential 5G techniques to fulfill unprecedent challenges posed by new requirements of quality of service (QoS) \cite{Tullberg2016METIS,wang2016throughput}. In particular, the provision of the ultra-reliability has attracted enormous interest in emerging communication systems. 
For instance, to accommodate many critical ultra-reliable and low-latency communications (uRLLC) services (e.g., industrial control and traffic safety), the outage probability is expected to be as low as $10^{-9}$ \cite{chen2018ultra}.
To this end, reliable communications for NOMA systems can be realized by adopting the retransmission mechanism \cite{choi2016re}. The representative retransmission technique is hybrid automatic repeat request (HARQ), which leverages the forward error correction (FEC) at the physical layer and the automatic repeat request (ARQ) at the link layer together\cite{caire2001throughput}. According to whether the retransmitted packets keep unchanged and which combining technique (diversity/code combining) is applied for decoding, HARQ scheme can be further divided into three basic categories, including Type I HARQ, HARQ with chase combining (HARQ-CC) and HARQ with incremental redundancy (HARQ-IR). Specifically, Type I HARQ utilizes the currently received packet for decoding, while HARQ-CC and HARQ-IR employ maximal ratio combining (MRC) and code combining for joint decoding with the erroneously received packets, respectively. To assure the reception reliability of new communication paradigms, it is of practical significance to evaluate the outage performance of different types of HARQ-aided NOMA systems.

%

The performance of HARQ-assisted NOMA systems has recently received considerable research attention in previous literature. More specifically, the system-level simulations were conducted to demonstrate potential gains of the NOMA scheme over the conventional OMA scheme by considering the link adaptation functionality of HARQ \cite{saito2013system}. 
Moreover, in order to ensure the improvement in the combined signal-to-interference-plus-noise ratio (SINR) in retransmissions, an opportunistic HARQ strategy was further proposed to assist NOMA in \cite{li2015investigation}. The simulated results disclosed that the opportunistic HARQ combination can bring NOMA gain improvement in cell throughput around 4\%$\sim$5\%. 
Moreover, once the near user succeeds to cancel out the NOMA interference from the far user, the interference-cancelled data can be recycled to support the far user through cooperative communications. For instance, D. Roh \emph{et al.} in \cite{roh2016improvement} assumed that the near user is capable of cooperating with the base station to resend the far user's signal with HARQ-CC in the presence of the decoding failure at the far user. The proposed cooperative retransmission method improves the outage probability, symbol error rate and throughput of the far user. In addition, the outage performance was further investigated for cooperative Type I HARQ-aided NOMA systems by accounting for the effect of aggregate network interference in \cite{zheng2018cooperative}. The proposed scheme attains an outage performance gain by 21\% compared to the non-cooperative one. In \cite{choi2016harq}, the concept of the error exponent was developed by using large deviations for the equivalent assessment of the outage performance of downlink NOMA-HARQ-IR systems. A lower bound was then derived for the error exponent. Based on this result, power allocation was carried out to guarantee a tolerable outage probability. Unfortunately, the analysis of the outage probability for HARQ-aided NOMA systems is rather involved because of the incorporation of multiple fractional random variables. Besides, most of the prior literature conduct the outage analyses for HARQ-aided NOMA systems based on either simulations or approximations \cite{saito2013system,li2015investigation,choi2016harq}. Particularly for HARQ-CC and HARQ-IR-aided NOMA schemes, the relevant outage probabilities are consequently transformed to 
distributions of the summation/product of multiple fractional random variables. Nevertheless, there are only several numerical approaches available to tackle these problems. To name a few, 
the outage performance of two-user HARQ-CC-assisted NOMA systems was approximated by using Gaver-Stehfest procedure in \cite{cai2018performance}. The analytical results then enabled the investigation of the tradeoff between retransmission times and power allocation coefficient. The similar methodology was afterwards employed to study the impact of time correlation upon the outage probability for two-user partial HARQ-assisted NOMA systems, where HARQ-CC and HARQ-IR are partially implemented in \cite{cai2019impact}. It was found in both \cite{cai2018performance} and \cite{cai2019impact} that full time diversity is achievable for HARQ-assisted NOMA systems by assuming perfect SIC, that is, the interfering signals can be completely eliminated when performing SIC. However, since the decoding failures for the farther user's message were ignored under the assumption of perfect SIC at the near user in \cite{cai2018performance,cai2019impact}, this may result in an overestimated performance of time diversity.

\subsection{Scope of the Paper}
In order to satisfy the stringent requirement of the reliability, the condition of high signal-to-noise ratio (SNR) is usually indispensable for the fulfillment of an exceedingly low target outage probability (e.g., less than $10^{-9}$). This motivates us to investigate the asymptotic characteristics of the outage probability for HARQ-aided NOMA systems in high SNR regime. In particular, the diversity order is a fundamental asymptotic reliability metric to characterize the degree of freedom for a communication system. More specifically, the diversity order describes how the outage probability scales with the transmit power under high SNR. Although it was proved in \cite{shi2016optimal,shi2017asymptotic} that full time diversity can be achieved by any type of the three conventional HARQ schemes, this result may not carry over to the HARQ-aided NOMA systems. Hence, a natural question arises here: \emph{Can full diversity be achieved by HARQ-aided NOMA systems?} To answer this crucial question, it is mandatory to study the diversity order for HARQ-aided NOMA systems. 

This paper thoroughly examines the achievable diversity orders for three types of HARQ-aided downlink NOMA systems over Rayleigh fading channels, including Type I HARQ, HARQ-CC and HARQ-IR schemes. First, the diversity order of two-user HARQ-aided NOMA systems is analyzed. Towards this goal, we have to obtain the outage probabilities for various HARQ-aided NOMA systems. Specifically, the outage probability of the Type I HARQ-aided NOMA system can be expressed in closed-form. However, it is nearly impossible to provide compact expressions for outage probabilities of the HARQ-CC and HARQ-IR-aided NOMA systems. This is due to the difficulty of dealing with the summation/product of multiple fractional random variables. Alternatively, the upper and lower bounds of the outage probabilities for the HARQ-CC and HARQ-IR-aided NOMA systems are got by recursively using the integration domain partition approach. Then combining the bounds with squeeze theorem yields the diversity order. Later, we go on to extend our results to the cases with an arbitrary number of users. Besides, the influences of power-efficient transmission strategy and imperfect channel state information (CSI) are investigated in depth. In the end, the analytical results are verified by undertaking Monte Carlo simulations. 

\subsection{Contributions}
The main contributions of this paper are listed as follows.
\begin{enumerate}
  \item The diversity order of HARQ-aided downlink NOMA systems is proved to be a decreasing step function of transmission rate given the ratio between the transmit powers allocated to the two users. Moreover, the analytical results reveal that HARQ-IR-aided NOMA systems have the largest diversity order, followed by HARQ-CC-aided and then Type I HARQ-aided NOMA systems.
  \item As opposed to the conventional HARQ systems that full time diversity is achievable irrespective of the value of transmission rate, HARQ-aided downlink NOMA systems are capable of attaining full diversity given the maximum allowable transmission rate below a threshold. The threshold of the transmission rate depends on HARQ type and the ratio between transmit powers.
  \item Likewise, the analytical results can be readily extended to a general case with more than two users. In addition, it is interesting to note that the users' diversity orders obey a descending order according to their respective average channel gains.
  \item We then propose an advanced power-efficient transmission strategy, which decreases the outage probability of the farther user. However, the resultant diversity order keeps unchanged as under the simple strategy. Besides, our analysis is applicable to the case of imperfect CSI.
\end{enumerate}

\subsection{Outline}
The remaining part of this paper is organized as follows. Section \ref{sec:sys_mod} presents the model of HARQ-aided NOMA systems. The diversity order is derived for various types of HARQ-aided NOMA systems in Section \ref{sec:div}. Section \ref{sec:num} validates the theoretical analysis. Eventually, some conclusions are outlined in the last section. Moreover, Table \ref{tab:list_symb} provides a list of notations used in the paper. 

\begin{table}[h]
  \centering
    \caption{List of Notations.}
\label{tab:list_symb}%
\begin{tabular}{p{1cm}||p{7cm}}
 \hline
 Notation& Definition\\
 \hline
  $\alpha_{i,k}$ & The channel gain between the source and user $i$ in the $k$-th HARQ round, i.e., $\alpha_{i,k} = |h_{i,k}|^2$.\\
 $\bar \alpha_i$ & The average channel gain between the source and user $i$.\\
   $c$, ${c_j}$, $c_1$ & $c=P_2/P_1$, ${c_j} = {{{P_j}}}/{{\sum\nolimits_{l = 1}^{j - 1} {{P_l}} }}$ and $c_1=\infty$.\\
  $d_i$, $\tilde d_i$ & The diversity order of user $i$, the approximate diversity order.\\
 $d_i^{\rm eff}$ & The diversity order of user $i$ under power-efficient transmission strategy.\\
  $d_{i\to j}$ & The diversity order associated with $\Pr \left\{ {{{I_{i \to j,K}} < {R_j}} } \right\}$.\\
    $E_\ell$& The event of the successful decoding of user $1$ after $\ell$ rounds.\\
  $\bar E$& The outage event at user 1.\\
   $\epsilon_{i,k}$  & Channel estimate error of $h_{i,k}$.\\
  $\delta$, $\Delta$& An power increment, an positive number in $0 < \Delta \le c$.\\
  $h_{i,k}$& Channel coefficient between the source and user $i$ in the $k$-th HARQ round.\\
  $\hat h_{i,k}$ & Channel estimate of $h_{i,k}$.\\
  $I_{i\to j,K}$ & The mutual information accumulated by user $i$ for decoding message $j$ after $K$ HARQ rounds.\\
 $I_{2\to 2,K,\ell}$ & The mutual information accumulated by user $2$ for decoding message $2$ after $K$ HARQ rounds conditioned on $E_\ell$.\\
 $K$& The maximum allowable number of transmissions of HARQ.\\
 $L$ & The length of the subcodeword.\\
 $M$& The number of users.\\
 $O(\cdot), \Theta(\cdot)$ & Big O notation, big Theta notation.\\
    $P_i$ & The average transmit power for the signal ${\bf x}_i$.\\
 $p_{i,K}^{out}$ & The outage probability of user $i$ after $K$ HARQ rounds.\\
 $p_{i,K}^{{\rm{eff}},out}$& The outage probability of user $i$ after $K$ HARQ rounds under power-efficient transmission strategy.\\
 $R_i$& The transmission rate for the signal ${\bf x}_i$.\\
 ${\bf w}_{i,k}$ & Additive Gaussian white noise with unity variance.\\
 ${\tilde{\bf w}_{i,k}}$ & Gaussian noise with variance $1+\sigma_i^2\sum\nolimits_{i=1}^MP_i$.\\
 ${\bf x}_i$& The signal intended for user $i$.\\
 ${{\bf y}_{i,k}}$& The received signal at user $i$ in the $k$-th HARQ round.\\
 $\lfloor \cdot \rfloor$, $\left[ \cdot \right]^+$ & Floor operator, the projection onto the nonnegative orthant.\\  
 \hline
\end{tabular}
\end{table}


\section{System Model}\label{sec:sys_mod}
A downlink NOMA system consisting of one source and multiple destination nodes with the aid of HARQ is considered in this paper. To ease understanding and subsequent analyses, we first examine the diversity order of the HARQ-aided NOMA system with only two destination users. Nevertheless, the similar analytical results can be readily extended to more general scenarios with an arbitrary number of deployed destination users, which will be briefly discussed in Subsection \ref{sec:ext}. To begin with, this section introduces the signal and channel models, the preliminary on HARQ-aided NOMA schemes and the definitions of outage probability and diversity order.
\subsection{Signal and Channel Models}
Denote by ${\bf x}_i$ the signal intended for destination user $i$, where $i=1,2$. The signals ${\bf x}_1$ and ${\bf x}_2$ are drawn from randomly generated Gaussian codebooks, and the transmission rate for ${\bf x}_i$ is denoted by $R_i$. To multiplex the two users at power-domain, the source employs NOMA to superimpose the signals ${\bf x}_1$ and ${\bf x}_2$ with different power levels. The average transmit powers allocated to the signals ${\bf x}_1$ and ${\bf x}_2$ are denoted by ${\mathbb E}(|{\bf x}_i|^2)=P_i$ for $i=1, 2$. 
To conserve the bandwidth and energy consumed on frequent feedbacks of instantaneous CSI, this paper considers that only the statistical CSI knowledge (i.e., the average channel gains) is known to the source. Meanwhile, the HARQ scheme is adopted to resend the superimposed message under bad channel condition to guarantee the transmission reliability. The maximum number of transmissions for each message is allowed up to $K$ to avoid network congestion. We denote the channel gain between the source and user $i$ in the $k$-th HARQ round by $\alpha_{i,k}$. By assuming independent Rayleigh fading channels, the channel gain $\alpha_{i,k}$ follows an exponential distribution with the probability density function (PDF) given by
\begin{equation}\label{eqn:alpha_dis}
f_{\alpha_{i,k}}(x) = \frac{1}{\bar \alpha_i}e^{-\frac{x}{\bar \alpha_i}},
\end{equation}
where ${\bar \alpha_i}$ stands for the average channel gain between the source and user $i$, i.e., ${\bar \alpha_i} = \mathbb E({\alpha_{i,k}})$. For simplicity of analysis, we assume that the channel gains are normalized such that the additive Gaussian white noise (AWGN) has unit variance.

\subsection{HARQ-Aided NOMA Schemes}
The benefit of using NOMA originates from the fact that the difference between fading channels can be exploited to substantially enhance the spectral efficiency. Without loss of generality, we assume that user 1 is closer to the source than user 2. In other words, user 2 has a larger path loss than user 1, i.e., ${\bar \alpha_1} > {\bar \alpha_2}$. Following the principle of NOMA, more transmission power (relative to the target transmission rate) will be allocated to the farther user (user 2) to ensure user fairness. As a result, user 1 has a lower SNR than user 2. After receiving the superimposed message, user 1 capitalizes on the SIC to subtract the interfering signal ${\bf x}_2$ before decoding its own message ${\bf x}_1$. However, user 2 directly decodes its own message ${\bf x}_2$ by treating user 1's message ${\bf x}_1$ as interference.

Due to the involvement of multiple transmissions into the HARQ-aided NOMA scheme, the message ${\bf x}_2$ might be recovered before ${\bf x}_1$ by both users\footnote{By considering the decoding order of SIC, the delivery of the message ${\bf x}_2$ is deemed to be a success if and only if both users succeed to decode ${\bf x}_2$.}. 
Clearly, it is unnecessary to resend the message ${\bf x}_2$ which becomes redundant in the following HARQ rounds. If positive acknowledgement (ACK) messages for ${\bf x}_2$ are fed back from both users to the source, ${\bf x}_2$ can be eliminated from the superposition coded signal in the succeeding retransmissions. In other words, only ${\bf x}_1$ needs to be retransmitted, which consequently yields the reduction of the power consumption. Nonetheless, since ${\bf x}_2$ is transparent to both users in the circumstances, the outage probability would not be influenced no matter whether ${\bf x}_1$ or the superposition coded signal is subsequently retransmitted. On the other hand, if user 1 successfully decodes its own message ${\bf x}_1$ prior to user 2, 
the natural idea is that only ${\bf x}_2$ needs to be resent instead of superposition coded signal. However, it will complicate the outage analysis due to the cumbersome transmission procedure. Thereafter, to simplify the system model as well as ease the analysis, the transmissions of superposition coded signal are therefore assumed in all HARQ rounds, which is also applicable to the scenario with an arbitrary number of NOMA users. In fact, this assumption can be treated as the worst case to some extent. Besides, it is believed that complicating the analysis by straightforwardly taking into account more complex transmission procedure will dilute the paper's contribution, but it will be briefly discussed in Subsection \ref{sec:power_eff}.

By applying different encoding and decoding operations at the transceiver, HARQ scheme can be classified into three types, including Type I HARQ, HARQ-CC and HARQ-IR. More specifically, the conventional Type I HARQ recovers the message by relying on the currently received packet, which amounts to using selection combining (SC). The erroneously received packets will be discarded directly without any necessity of saving them. However, it undoubtedly yields the waste of the useful information embedded in these packets. To combat this issue, this paper assumes the HARQ-aided NOMA system is equipped with a dedicated buffer to store the erroneous packets for subsequent decodings. The buffer not only supports the employment of SC for Type I HARQ-aided NOMA system, but also is necessary to store the messages that should be decoded prior to its own message. 
To further improve the outage performance, HARQ-CC and HARQ-IR adopt more advanced combining techniques at the transceiver. In particular, the techniques of MRC and code combining are applied to HARQ-CC and HARQ-IR, respectively. It is worth mentioning that HARQ-IR surpasses both Type I HARQ and HARQ-CC in terms of the outage probability notwithstanding the extra complexity cost\footnote{It is noteworthy that the analysis of the computational complexity not only depends on the type of HARQ scheme but also the kinds of FEC code, where the FEC code could be Reed-Solomon code, convolutional code, turbo code, polar code, etc. The computational complexity of HARQ schemes has been reported in \cite{choi2001class,chen2013survey,chen2013hybrid,chen2014polar}. By taking the polar coded HARQ schemes as examples, the decoding operations of Type I HARQ, HARQ-CC and HARQ-IR schemes at the $K$-th transmission require complexities on the order of $O(L\log L)$, $O(L(K+\log L))$ and $O(KL\log(KL) )$, respectively, where $O(\cdot)$ refers to the big-O notation and $L$ denotes the length of the subcodeword transmitted in each HARQ round \cite{trifonov2012efficient}.}. In addition, the outage performance of HARQ-CC and -IR schemes essentially benefits from the combining techniques rather than the equipment of buffer. Nevertheless, the buffer is an indispensable part of HARQ-CC and -IR.
\subsection{Outage Probability and Diversity Order}
From the information-theoretical perspective, the accumulated mutual information achieved by user $i$ for decoding ${\bf x}_2$ after $K$ HARQ rounds can be expressed as \cite{caire2001throughput,shi2018energy}
\begin{align}\label{eqn:noma_harq_inf_i2}
&{I_{i \to 2,K}} = \notag\\
&\left\{ {\begin{array}{*{20}{c}}
{\max {{\left\{ {{{\log }_2}\left( {1 + \frac{{{\alpha _{i,k}}{P_2}}}{{{\alpha _{i,k}}{P_1} + 1}}} \right)}:{k \in \left[ {1,K} \right]} \right\}}},}&{\rm{I}}\\
{{{\log }_2}\left( {1 + \sum\nolimits_{k = 1}^K {\frac{{{\alpha _{i,k}}{P_2}}}{{{\alpha _{i,k}}{P_1} + 1}}} } \right),}&{{\rm{CC}}}\\
{\sum\nolimits_{k = 1}^K {{{\log }_2}\left( {1 + \frac{{{\alpha _{i,k}}{P_2}}}{{{\alpha _{i,k}}{P_1} + 1}}} \right)}, }&{{\rm{IR}}}
\end{array}} \right.,
\end{align}
where the occurrence of the term ${\alpha _{i,k}}{P_1}$ in the denominator is due to the fact that user 1's message is treated as Gaussian noise whilst decoding ${\bf x}_2$.

In addition, only if user 1 successfully reconstructs user 2's message such that ${I_{1 \to 2,K}} \ge R_2$, user 1 is able to decode its own message ${\bf x}_1$ by perfectly eliminating the interfering signal ${{\bf x}_2}$ via SIC. Accordingly, the accumulated mutual information achieved by user 1 for decoding its own message after $K$ HARQ rounds is given by
\begin{align}\label{eqn:noma_harq_inf_1}
&{I_{1\to 1,K}} = \notag\\
&\left\{ {\begin{array}{*{20}{c}}
{\max \left\{ {{{\log }_2}\left( {1 + {\alpha _{1,k}}{P_1}} \right):k \in \left[ {1,K} \right]} \right\},}&{\rm{I}}\\
{{{\log }_2}\left( {1 + \sum\nolimits_{k = 1}^K {{\alpha _{1,k}}{P_1}} } \right),}&{{\rm{CC}}}\\
{\sum\nolimits_{k = 1}^K {{{\log }_2}\left( {1 + {\alpha _{1,k}}{P_1}} \right)} ,}&{{\rm{IR}}}
\end{array}} \right..
\end{align}

By using the Gaussian codes and typical-set decoding, the outage event occurs at user 1 after $K$ HARQ rounds if user 1 fails to either subtract the interfering signal ${\bf x}_2$ or recover its own message. On the basis of \eqref{eqn:noma_harq_inf_i2} and \eqref{eqn:noma_harq_inf_1}, the outage probability of user 1, $p_{1,K}^{out}$, can be obtained as
\begin{equation}\label{eqn:out_user1}
p_{1,K}^{out} = \Pr \left\{ {{{I_{1\to 1,K}}} < {R_1}\bigcup {{I_{1 \to 2,K}} < {R_2}} } \right\}.
\end{equation}
Likewise, the outage event happens at user 2 if and only if user 2 fails to decode its own message. Hence, the outage probability of user 2, $p_{2,K}^{out}$, can be written as
\begin{equation}\label{eqn:out_user2}
p_{2,K}^{out} = \Pr \left\{ {{{I_{2 \to 2,K}} < {R_2}} } \right\}.
\end{equation}
Unfortunately, it is generally intractable to derive closed-form expressions for the outage probabilities of HARQ-aided NOMA systems due to the incorporation of multiple fractional random variables. Instead, we turn to examine the asymptotic characteristics of the outage performance. In this paper, we are interested in the diversity order of HARQ-aided NOMA systems. Specifically, the diversity order is an important performance metric to characterize the asymptotic scaling law of the outage probability with respect to the average transmit SNR. 
More precisely, the diversity order associated with user $i$ is explicitly given by\footnote{\label{ft:def_d}In the numerical analysis, we can define $\tilde d_i$ as \[{\tilde d_i} =  10\frac{{{{\log }_{10}}p_{i,K}^{out}\left( {{{[{P_i}]}_{{\rm{dB}}}}} \right) - {{\log }_{10}}p_{i,K}^{out}\left( {{{[{P_i}]}_{{\rm{dB}}}} + {{\left[ \delta  \right]}_{{\rm{dB}}}}} \right)}}{{{{\left[ \delta  \right]}_{{\rm{dB}}}}}}\] and use it to approximate the diversity order $d$, where $\delta$ is a power increment, $p_{i,K}^{out}\left( {{{[{P_i}]}_{{\rm{dB}}}}} \right)$ denotes the outage probability given ${P_i}$ in decibels. This is due to the fact that ${\tilde d_i}$ commonly has a faster convergence than the definition in \eqref{eqn:diver_order_def} as $P_i$ approaches to infinity.}
\begin{equation}\label{eqn:diver_order_def}
   d_i = -\mathop {\lim }\limits_{P_i  \to \infty } \frac{{\log {p_{i,K}^{out}}}}{{\log P_i }}.
\end{equation}
From \eqref{eqn:diver_order_def}, the diversity order quantifies the decreasing slope of the outage curves against the average transmit SNR on a log-log scale. More specifically, the outage probability at high SNR behaves asymptotically as $p_{i,K}^{out} \simeq {\left( \mathcal C  \cdot P_i  \right)^{ - d_i }}$ \cite[eq.(3.158)]{tse2005fundamentals}, \cite[eq.(1)]{wang2003simple}, where $\mathcal C $ is termed as the coding gain. Although the authors in \cite{shi2016optimal,shi2017asymptotic} proved that the three conventional HARQ schemes can achieve the same and full diversity order, this result is inapplicable to the HARQ-aided NOMA systems. Therefore, it is mandatory to study the achievable diversity order of the HARQ-aided NOMA systems. Moreover, to simplify the analysis as well as ensure the user fairness, as the transmit powers increase, the ratio between them is assumed to be fixed such that $P_2 = c P_1$.
\section{Analysis of Diversity Order}\label{sec:div}
The diversity order is studied for user 2 first due to the simple form of its outage probability. The analytical results are then applied to derive the diversity order of user 1.
\subsection{Diversity Order of User 2}\label{sec:diver_order_user2}
Recalling that different types of HARQ schemes yield distinct expressions of accumulated mutual information, it is necessary to analyze the diversity order of user 2 for each type of HARQ-aided NOMA scheme individually.
\subsubsection{Type I HARQ}
If Type I HARQ is applied to assist the NOMA transmission, the outage probability of user 1 can be obtained by substituting \eqref{eqn:noma_harq_inf_i2} into \eqref{eqn:out_user2} as
\begin{align}\label{eqn:outage_typeI}
{p_{2,K}^{out,I}} &= \Pr \left\{ {\bigcap\nolimits_{k = 1}^K {{{\log }_2}\left( {1 + \frac{{{\alpha _{2,k}}{P_2}}}{{{\alpha _{2,k}}{P_1} + 1}}} \right) < {R_2}} } \right\}\notag\\
&= \prod\nolimits_{k = 1}^K {{{{\Pr \left\{ {{\alpha _{2,k}} < \frac{{{2^{{R_2}}} - 1}}{{{P_2} - \left( {{2^{{R_2}}} - 1} \right){P_1}}}} \right\}} }}}, 
\end{align}
where the first step holds by using the independence between fading channels across all HARQ rounds. By putting \eqref{eqn:alpha_dis} into \eqref{eqn:outage_typeI}, ${p_{2,K}^{out,I}}$ can be expressed in closed-form as
\begin{align}\label{eqn:outage_typeI_fin}
&{p_{2,K}^{out,I}}= \notag\\
&\left\{ {\begin{array}{*{20}{c}}
{{{\left( {1 - {e^{ - \frac{{{2^{{R_2}}} - 1}}{{\bar \alpha_2 \left( {{P_2} - \left( {{2^{{R_2}}} - 1} \right){P_1}} \right)}}}}} \right)}^K},}&{\frac{{{2^{{R_2}}} - 1}}{c} < 1}\\
{1,}&{\rm else}
\end{array}} \right..
\end{align}

By plugging \eqref{eqn:outage_typeI_fin} into \eqref{eqn:diver_order_def} together with the Maclaurin series of the exponential function, we have 
\begin{equation}\label{eqn:diver_order_I}
  d_2^{I} = K\left[1-\left\lfloor {\frac{2^{R_2}-1}{c}} \right\rfloor \right]^+,\,
\end{equation}
where $\lfloor \cdot \rfloor$ is the floor operator and $\left[ \cdot \right]^+$ denotes the projection onto the nonnegative
orthant such that $\left[ x \right]^+ = \max\{x,0\}$. It is shown in \eqref{eqn:diver_order_I} that the diversity order equals to $K$ if ${2^{R_2}-1}<{c}$ and equals to zero otherwise.

\subsubsection{HARQ-CC}
Similarly, the outage probability of user 2 for HARQ-CC-aided NOMA is rewritten by inserting \eqref{eqn:noma_harq_inf_i2} into \eqref{eqn:out_user2} as
\begin{align}\label{eqn:outage_HARQ-CC}
{p_{2,K}^{out,CC}} 
 &= \Pr \left\{ {\underbrace {\sum\nolimits_{k = 1}^K {\frac{{{\alpha _{2,k}}{P_2}}}{{{\alpha _{2,k}}{P_1} + 1}}} }_{\gamma}  < {2^{R_2}} - 1} \right\}.
\end{align}
It thus boils down to determining the distribution of a summation of multiple fractional random variables, which extremely impedes the derivation of the compact expression for ${p_{2,K}^{out,CC}}$. 
Hence, simple lower and upper bounds of the outage probability are developed to ease the analysis of the diversity order. To do so, the method of moment generating function (MGF) is applied here to derive the distribution of $\gamma$. The MGF of $\gamma$ is precisely given by
\begin{align}\label{eqn:out_cc_mgf}
&{\mathbb{E}}\left\{ {{e^{t\gamma }}} \right\} = \prod\nolimits_{k = 1}^K {{\mathbb{E}}\left\{ {{e^{\frac{{{\alpha _{2,k}}{P_2}}}{{{\alpha _{2,k}}{P_1} + 1}}t}}} \right\}} \notag\\
 &= \prod\nolimits_{k = 1}^K {\frac{1}{{\bar \alpha_2 }}\int\nolimits_0^\infty  {{e^{\frac{{x_k{P_2}}}{{x_k{P_1} + 1}}t}}{e^{ - \frac{x_k}{{\bar \alpha_2 }}}}dx_k} }  \notag\\
 &= \prod\nolimits_{k = 1}^K \frac{{c{e^{ct + \frac{1}{{\bar \alpha_2 {P_1}}}}}}}{{{P_1{\bar \alpha }_2}}}\int\nolimits_{0}^{c} {{e^{ - {t}{z_k} - \frac{c}{{{P_1{\bar \alpha }_2}}}\frac{1}{z_k}}}\frac{1}{{z_k}^2}dz_k},
\end{align}
where the last step holds by making the change of variable $z_k = c/(x_kP_1+1)$. \eqref{eqn:out_cc_mgf} can be further written in closed-form by using the definition of the upper incomplete Fox's H function in \cite{yilmaz2009productshifted}. The result is omitted here due to space limitations. 
Furthermore, the cumulative distribution function (CDF) of $\gamma$ is obtained by using inverse Laplace transform as
\begin{align}\label{eqn:cdf_gamma_sum}
{F_\gamma }\left( \gamma  \right) &= \frac{1}{2\pi \rm i}\int_{a-{\rm i}\infty}^{a+{\rm i}\infty} {\frac{{{e^{t\gamma }}}}{t}} {\mathbb{E}}\left\{ {{e^{-t\gamma }}} \right\} dt.
\end{align}
where $a > 0$ and ${\rm i} = \sqrt{-1}$. The representation of the contour integral in (\ref{eqn:cdf_gamma_sum}) can be numerically evaluated by using the popular software packages, such as Mathematica and Matlab. By substituting \eqref{eqn:out_cc_mgf} into \eqref{eqn:cdf_gamma_sum} and interchanging the order of integrations, it follows that
\begin{align}\label{eqn:cdf_gamma_sum1}
{F_\gamma }\left( \gamma  \right)
&= {\left( {\frac{c}{{{P_1}{{\bar \alpha }_2}}}} \right)^K}{e^{\frac{K}{{{{\bar \alpha }_2}{P_1}}}}}\notag\\
&\times \int_0^c { \cdots \int_0^c {{e^{ - \frac{c}{{{P_1}{{\bar \alpha }_2}}}\sum\nolimits_{k = 1}^K {\frac{1}{{{z_k}}}} }}\prod\nolimits_{k = 1}^K {\frac{1}{{{z_k}^2}}} d{z_1} \cdots d{z_K}}}\notag\\
  &\times \frac{1}{{2\pi \rm i}}\int_{a - {\rm i}\infty }^{a + {\rm i}\infty } {\frac{1}{t}{e^{t\left( {\gamma  + \sum\nolimits_{k = 1}^K {{z_k}}  - Kc} \right)}}} dt.
\end{align}

By using the inverse Laplace transform of the unit step function, the CDF ${F_\gamma }\left( \gamma  \right)$ can be rewritten as \eqref{eqn:cdf_gamma_sum2}, as shown at the top of the next page, where $u(x)$ denotes the unit step function.
\begin{figure*}[!t]
\begin{align}\label{eqn:cdf_gamma_sum2}
{F_\gamma }\left( \gamma  \right)  = {\left( {\frac{c}{{{P_1}{{\bar \alpha }_2}}}} \right)^K}{e^{\frac{K}{{{{\bar \alpha }_2}{P_1}}}}}\underbrace {\int_0^c { \cdots \int_0^c {{e^{ - \frac{c}{{{P_1}{{\bar \alpha }_2}}}\sum\nolimits_{k = 1}^K {\frac{1}{{{z_k}}}} }}u\left( {\gamma  + \sum\nolimits_{k = 1}^K {{z_k}}  - Kc} \right)\prod\nolimits_{k = 1}^K {\frac{1}{{{z_k}^2}}} d{z_1} \cdots d{z_K}} } }_{{\phi _K}\left( \gamma  \right)}.
\end{align}
\hrulefill
\end{figure*}
Unfortunately, it is nearly impossible to obtain a closed-form expression for ${F_\gamma }\left( \gamma  \right)$ due to the representation of the intractable multiple integral of ${\phi _K}\left( \gamma  \right)$.
By plugging \eqref{eqn:cdf_gamma_sum2} into \eqref{eqn:outage_HARQ-CC}, the outage probability of HARQ-CC-aided NOMA system is rewritten as
\begin{equation}\label{eqn:out_prob_cc_phi_xp}
{p_{2,K}^{out,CC}} = {\left( {\frac{c}{{{P_1}{{\bar \alpha }_2}}}} \right)^K}{e^{\frac{K}{{{{\bar \alpha }_2}{P_1}}}}}{{\phi _K}\left( 2^{R_2}-1  \right)}.
\end{equation}
 To extract further insights on the outage probability as well as obtain the diversity order, it obliges us to study the asymptotic behaviour of ${\phi _K}\left( \gamma  \right)$ in high SNR regime. On the basis of \eqref{eqn:cdf_gamma_sum2}, ${\phi _K}\left( \gamma  \right)$ is proved in Appendix \ref{app:harq_cc_boundphi} to be lower and upper bounded by applying integration domain partition trick as
%
%
\begin{multline}\label{eqn:g_red_fin}
\frac{{{{\bar \alpha }_2}{P_1}}}{c}{e^{ - \frac{1}{{\bar \alpha_2 {P_1}}}}}{\phi _{K - 1}}(\gamma  - c) \le {\phi _K}\left( \gamma  \right)\le {e^{ - \frac{1}{{{{\bar \alpha }_2}{P_1}}}}}\\
 \times\left( {\frac{{c - \Delta }}{{c\Delta }}{\phi _{K - 1}}(\gamma ) + \frac{{{{\bar \alpha }_2}{P_1}}}{c}{\phi _{K - 1}}\left( {\gamma  - c + \Delta } \right)} \right),
\end{multline}
where $\Delta$ could be any positive number smaller than or equal to $c$, i.e., $0 < \Delta \le c$. In fact, \eqref{eqn:g_red_fin} uncovers a recursive relationship for ${\phi _K}\left( \gamma  \right)$. By combining \eqref{eqn:g_red_fin} with the following lemma concerning the asymptotic behaviour of ${\phi _K}\left( \gamma  \right)$ as $P_1$ approaches to $\infty$, the diversity order of user 2 for HARQ-CC-aided NOMA system can be obtained.
\begin{lemma}\label{the:cc_phi}
${\phi _K}\left( \gamma  \right)$ is an increasing function of $P_1$ and $\gamma$. If $0<\gamma < c$, the lower and upper bounds of ${\phi _K}\left( \gamma  \right)$ are given by
\begin{multline}\label{eqn:phi_K_cc_bounds_c1}
{e^{ - \frac{{Kc}}{{\bar \alpha_2 {P_1}\left( {c - \frac{\gamma }{K}} \right)}}}}{\left( {\frac{1}{{c - {\gamma }/{K}}} - \frac{1}{c}} \right)^K}
\le {\phi _K}\left( \gamma  \right)  \\ \le {e^{ - \frac{K}{{{{\bar \alpha }_2}{P_1}}}}}{\left( {\frac{1}{{c - \gamma }} - \frac{1}{c}} \right)^K}.
\end{multline}
Additionally, ${\phi _K}\left( \gamma  \right) = {{\left( {{{\bar \alpha_2 }{P_1}}/{c}} \right)}^K}{\exp\left({ - {K}/({{\bar \alpha_2 {P_1}}})}\right)}$ for $\gamma \ge Kc$. As $P_1$ approaches to $\infty$, ${\phi _K}\left( \gamma  \right)$ is asymptotic to 
\begin{equation}\label{eqn:phi_K_gamma_ge}
 {\phi _K}\left( \gamma  \right) = \left\{ {\begin{array}{*{20}{c}}
0,&{\gamma  = 0}\\
{\Theta(1)\ne 0},&{0 < \gamma  < c}\\
{\Theta(P_1^k)},&{kc \le \gamma  < (k+1)c}\\
\Theta(P_1^K),&{\gamma  \ge Kc}
\end{array}} \right.,
\end{equation}
where $k\in[1,K-1]$, $\Theta(\cdot)$ denotes the big-Theta notation and we say $f(x)= \Theta(g(x))$ if there exist positive constants $k_1$, $k_2$ and $x_0$ such that $k_1 g(x) \le f(x) \le k_2 g(x)$ for all $x>x_0$. 
\end{lemma}
\begin{proof}
  Please refer to Appendix \ref{app:proof_phi_gener}.
\end{proof}

By combining \eqref{eqn:out_prob_cc_phi_xp}, \eqref{eqn:g_red_fin} and \eqref{eqn:phi_K_cc_bounds_c1}, the lower and upper bounds of the outage probability ${p_{2,K}^{out,CC}}$ can be calculated. More importantly, by substituting \eqref{eqn:out_prob_cc_phi_xp} into \eqref{eqn:diver_order_def} and using Lemma \ref{the:cc_phi}, the diversity order of user 2 for HARQ-CC-aided NOMA systems is given by
\begin{equation}\label{eqn:diversity_order_cc}
d_2^{CC} = {\left[ {K - \left\lfloor {\frac{{{2^{R_2}} - 1}}{c}} \right\rfloor } \right]^ + }.
\end{equation}

\subsubsection{HARQ-IR}
The outage probability of user 2 for HARQ-IR-aided NOMA system is expressed by inserting \eqref{eqn:noma_harq_inf_i2} into \eqref{eqn:out_user2} as
\begin{align}\label{eqn:cdf_gamma_sum_ir}
{p_{2,K}^{out,IR}} = \Pr \left\{ {\underbrace{\prod\nolimits_{k = 1}^K {\left( {1 + \frac{{{\alpha _{2,k}}{P_2}}}{{{\alpha _{2,k}}{P_1} + 1}}} \right)}}_{\tilde \gamma}  < {2^{{R_2}}}} \right\}.
\end{align}
Hence, ${p_{2,K}^{out,IR}}$ is determined by the distribution of the product of multiple shifted fractional random variables, i.e., ${\tilde \gamma}$. Similarly to \eqref{eqn:cdf_gamma_sum2}, the Mellin transform can be applied to rewrite ${p_{2,K}^{out,IR}}$ into \eqref{eqn:cdf_gamma_sum_ir_fin}, as shown at the top of the next page,
\begin{figure*}[!t]
\begin{align}\label{eqn:cdf_gamma_sum_ir_fin}
{F_{\tilde \gamma }}\left( \gamma  \right) = {\left( {\frac{c}{{{{\bar \alpha }_2}{P_1}}}} \right)^K}{e^{\frac{K}{{{{\bar \alpha }_2}{P_1}}}}}\underbrace {\int\nolimits_0^c { \cdots \int\nolimits_0^c {u\left( {\gamma  - \prod\nolimits_{k = 1}^K {\left( {1 + c - {z_k}} \right)} } \right){e^{ - \frac{c}{{{{\bar \alpha }_2}{P_1}}}\sum\nolimits_{k = 1}^K {\frac{1}{{{z_k}}}} }}\prod\nolimits_{k = 1}^K {\frac{1}{{{z_k}^2}}} d{z_1} \cdots d{z_K}} } }_{{\psi _K}\left( \gamma  \right)},
\end{align}
\hrulefill
\end{figure*}
where the proof is detailed in Appendix \ref{app:proof_ir_out}.
By substituting \eqref{eqn:cdf_gamma_sum_ir_fin} into \eqref{eqn:cdf_gamma_sum_ir}, the outage probability of HARQ-IR-aided NOMA system is given by
\begin{equation}\label{eqn:out_prob_ir_phi_xp}
{p_{2,K}^{out,IR}} = {\left( {\frac{c}{{{P_1}{{\bar \alpha }_2}}}} \right)^K}{e^{\frac{K}{{{{\bar \alpha }_2}{P_1}}}}}{{\psi _K}\left( 2^{R_2}  \right)}.
\end{equation}
Likewise, it is vital to investigate the asymptotic characteristics of ${{\psi _K}\left( \gamma \right)}$ to gain helpful insights into the outage probability under high SNR. By applying integration domain partition approach to \eqref{eqn:cdf_gamma_sum_ir_fin}, the lower and upper bounds of ${\psi _K}\left( {{\gamma}} \right)$ are obtained as
\begin{align}\label{eqn:psi_lower_upper}
 \frac{{\bar \alpha_2 {P_1}}}{c}{e^{ - \frac{1}{{\bar \alpha_2 {P_1}}}}}{\psi _{K - 1}}\left( {\frac{{{\gamma}}}{{1 + c}}} \right) \le {\psi _K}\left( {{\gamma}} \right) \le {e^{ - \frac{1}{{{{\bar \alpha }_2}{P_1}}}}}\notag\\
 \times \left( {\frac{{c - \Delta }}{{c\Delta }}{\psi _{K - 1}}\left( \gamma  \right) + \frac{{\bar \alpha_2 {P_1}}}{c}{\psi _{K - 1}}\left( {\frac{\gamma }{{1 + c - \Delta }}} \right)} \right),
\end{align}
where the proof is provided in Appendix \ref{app:psi_lb}. Similarly to \eqref{eqn:g_red_fin}, \eqref{eqn:psi_lower_upper} also presents a recursive relationship for ${\psi _K}\left( {{\gamma}} \right)$.

\begin{lemma}\label{the:psi_asy}
${\psi _K}\left( {{\gamma}} \right)$ is an increasing function of $P_1$ and $\gamma$. If $1<\gamma<1+c$, the lower and upper bounds of ${\psi _K}\left( {{\gamma}} \right)$ are given by
\begin{align}\label{eqn:phi_K_lowerupper}
{\left( {\frac{1}{{1+c-\sqrt[K]{\gamma}}} - \frac{1}{c}} \right)^K}{e^{ - \frac{{Kc}}{{\bar \alpha_2 {P_1}\left( {1+c-\sqrt[K]{\gamma}} \right)}}}} \le {\psi _K}\left( \gamma \right) \notag\\
\le e^{ - \frac{K}{\bar \alpha_2 {P_1}} }{\left( {\frac{1}{{1+c - \gamma }} - \frac{1}{c}} \right)^K},
\end{align}
Additionally, ${\psi _K}\left( {{\gamma}} \right) = {{{\left( {{{\bar \alpha_2 }}{P_1}/c} \right)}^K}{\exp{\left( - {K}/{\left({\bar \alpha_2 {P_1}}\right)}\right)}}}$ for $\gamma  \ge (1+c)^K$. As $P_1$ approaches to $\infty$, ${\psi _K}\left( {{\gamma}} \right)$ can be asymptotically expressed as  
\begin{align}\label{eqn:psi_K_gamma_ge}
&{\psi _K}\left( {{\gamma}} \right)= \notag\\
& \left\{ {\begin{array}{*{20}{c}}
0,&{{\gamma} = 1}\\
{{{\Theta(1)\ne 0}},}&{1 < {\gamma} < 1 + c}\\
{\Theta(P_1^k),}&{{{\left( {1 + c} \right)}^k} \le {\gamma} < {{\left( {1 + c} \right)}^{k+1}}}\\
{\Theta(P_1^K),}&{{\gamma} \ge {{\left( {1 + c} \right)}^K}}
\end{array}} \right.,
\end{align}
where $k\in[1,K-1]$.
\end{lemma}
\begin{proof}
Please see Appendix \ref{app:proof_psi_gener}.
\end{proof}

By combining \eqref{eqn:out_prob_ir_phi_xp}, \eqref{eqn:psi_lower_upper} and \eqref{eqn:phi_K_lowerupper}, the calculations for the lower and upper bounds of ${p_{2,K}^{out,IR}}$ are enabled. By putting \eqref{eqn:out_prob_ir_phi_xp} into \eqref{eqn:diver_order_def} and using Lemma \ref{the:psi_asy}, the diversity order of user 2 for HARQ-IR-aided NOMA systems is obtained by
\begin{equation}\label{eqn:diversity_order_ir}
d_2^{IR} = {\left[ {K - \left\lfloor {\frac{R_2}{\log_2(1+c)}} \right\rfloor } \right]^ + }.
\end{equation}

\subsection{Diversity Order of User 1}
According to \eqref{eqn:out_user1}, it is almost intractable to obtain a closed-form expression for the outage probability of user 1 due to the correlation between ${I_{1\to 1,K}}$ and ${I_{1 \to 2,K}}$. Although the outage probability $p_{1,K}^{out}$ can not be easily derived, the associated diversity order can be obtained by using the following lemma.
\begin{lemma}\label{the:mul_events_diver}
If the probability of the event $A_j$ follows the asymptotic behaviour as $\Pr\{A_j\}=\Theta(P^{-d_{A_j}})$ as $P\to P_0$ and $j\in [1,J]$, the probability of the union of the events $\bigcup\nolimits_{j = 1}^J {{A_j}} $ is asymptotic to
  \begin{equation}\label{eqn:prob_union_asy}
   \Pr\left\{\bigcup\nolimits_{j = 1}^J {{A_j}}\right\} = \Theta(P^{-\min\{d_{A_j}:j\in [1,J]\}}).
  \end{equation}
We denote by $d_{\bigcup }$ the diversity order with regard to the union of events $A_j$ for $j=1,\cdots,J$, and $d_{\bigcup }= \min\{d_{A_j}:j\in [1,J]\}$.
\end{lemma}
\begin{proof}
By using the inclusion-exclusion principle, $\Pr\left\{\bigcup\nolimits_{j = 1}^J {{A_j}}\right\}$ is bounded as
\begin{align}\label{eqn:prob_union_inex}
 \max\{\Pr\left({{A_j}}\right):j\in [1,J]\} &\le \Pr\left\{\bigcup\nolimits_{j = 1}^J {{A_j}}\right\} \le \sum\nolimits_{j = 1}^J {\Pr\left({{A_j}}\right)}.
\end{align}
By applying squeeze theorem to \eqref{eqn:prob_union_inex}, the proof is completed.
\end{proof}

By applying Lemma \ref{the:mul_events_diver} to \eqref{eqn:out_user1}, we conclude that the diversity order of user 1 is determined by the diversity orders associated with $\Pr \left\{ {{{I_{1\to 1,K}}} < {R_1}} \right\}$ and $\Pr \left\{ {{I_{1 \to 2,K}} < {R_2}} \right\}$. Evidently, the asymptotic behavior of $\Pr \left\{ {{I_{1 \to 2,K}} < {R_2}} \right\}$ is the same as that of the outage probability of user 2 because the accumulated mutual information ${I_{1 \to 2,K}}$ has the similar form as ${I_{2 \to 2,K}}$. Hence, it follows that $\Pr \left\{ {{I_{1 \to 2,K}} < {R_2}} \right\}=\Theta({P_1}^{-d_2})$. In addition, with \eqref{eqn:noma_harq_inf_1}, the probabilities $\Pr \left\{ {{{I_{1\to 1,K}}} < {R_1}} \right\}$ for Type I HARQ, HARQ-CC and HARQ-IR are expressed as the distributions of the maximum, summation and product of $K$ independent exponential random variables, respectively, which degenerate to the key problems of the conventional HARQ schemes. Fortunately, it has been proved in \cite{shi2016optimal,shi2017asymptotic} that full diversity can be achieved by the conventional HARQ systems, that is, the diversity order is equal to $K$. Therefore, Lemma \ref{the:mul_events_diver} indicates $d_1 = \min\{d_2,K\}=d_2$ owing to $d_2 \le K$. Accordingly, the diversity order of the HARQ-aided NOMA system is restricted by the user with the first decoding order (the one with the worse channel quality), namely user 2. 
\subsection{Discussions}
\subsubsection{Comparison between the Three Types of HARQ-aided NOMA systems}
By comparing \eqref{eqn:diver_order_I} and \eqref{eqn:diversity_order_cc}, it is readily found that the diversity order of HARQ-CC-aided NOMA system is greater than that of Type I HARQ-aided NOMA system, i.e.,
\begin{equation}\label{eqn:d_I_cc_rel}
  d_2^I \le d_2^{CC}.
\end{equation}
Then by comparing \eqref{eqn:diversity_order_cc} to \eqref{eqn:diversity_order_ir} together with the finding that $\left({2^R-1}\right)/{c}\le n$ is a subset of ${R}/{\log_2(1+c)}\le n$ for $n\ge1$ because of $(1+c)^n\ge 1+cn$, we arrive at $\lfloor \left({2^R-1}\right)/{c} \rfloor \ge \lfloor {R}/{\log_2(1+c)} \rfloor$. By combining \eqref{eqn:diversity_order_cc} and \eqref{eqn:diversity_order_ir}, the relationship between $d_2^{CC}$ and $d_2^{IR}$ obeys 
\begin{equation}\label{eqn:d_cc_IR_rel}
  d_2^{CC} \le d_2^{IR}.
\end{equation}
From \eqref{eqn:d_I_cc_rel} and \eqref{eqn:d_cc_IR_rel}, it is found that HARQ-IR-aided NOMA system performs the best in terms of the diversity order. Whereas, the Type I HARQ-aided NOMA system has the lowest diversity order and HARQ-CC performs in between them. This is not beyond our expectation because HARQ-IR achieves the best performance at the price of its high computation complexity, while Type I HARQ and HARQ-CC attempt to moderately reduce the computation complexity of encodings and decodings for the sake of saving the costs of hardware and energy. The similar results as \eqref{eqn:d_I_cc_rel} and \eqref{eqn:d_cc_IR_rel} also apply to the relationship between the diversity orders of user 1 under different HARQ schemes, namely $d_1^I \le d_1^{CC} \le d_1^{IR}$.
\subsubsection{Extension to Multiple NOMA Users}\label{sec:ext}
Our analytical findings can be extended to more general scenarios with more than two users. To proceed, we assume that there are $M$ NOMA users. The definitions for the notations ${\bf x}_{i}$, $\alpha_{i,k}$, $\bar \alpha_i$, $R_i$ and $P_i$ in Section \ref{sec:sys_mod} are also applicable here, where $i\in[1,M]$ and $k\in[1,K]$. Suppose that the average channel gains are sorted as $\bar \alpha_1\ge \bar \alpha_2 \ge \cdots \ge \bar \alpha_M$, the NOMA principle suggests that the decoding order of users conforms to the ascending order of the average channel gains. The outage event occurs at user $i$ if user $i$ fails to decode any user $j$'s message ${\bf x}_{j}$ for $i \le j \le M$. By applying Gaussian codes and typical-set decoding, the failure of decoding user $j$'s message takes place if the accumulated mutual information after $K$ HARQ rounds is less than the preset target transmission rate for user $j$, i.e., ${I_{i \to j,K}} < R_j$. Similarly to \eqref{eqn:noma_harq_inf_i2}, the mutual information after $K$ HARQ rounds accumulated by user $i$ intended for decoding ${\bf x}_{j}$ is written as
\begin{align}\label{eqn:noma_harq_inf_i2m2}
&{I_{i \to j,K}} = \notag\\
&\left\{ {\begin{array}{*{20}{c}}
{\max \left\{ {{{\log }_2}\left( {1 + \frac{{{\alpha _{i,k}}{P_j}}}{{{\alpha _{i,k}}\sum\nolimits_{l = 1}^{j - 1} {{P_l}}  + 1}}} \right):k \in \left[ {1,K} \right]} \right\},}&{\rm{I}}\\
{{{\log }_2}\left( {1 + \sum\nolimits_{k = 1}^K {\frac{{{\alpha _{i,k}}{P_j}}}{{{\alpha _{i,k}}\sum\nolimits_{l = 1}^{j - 1} {{P_l}}  + 1}}} } \right),}&{{\rm{CC}}}\\
{\sum\nolimits_{k = 1}^K {{{\log }_2}\left( {1 + \frac{{{\alpha _{i,k}}{P_j}}}{{{\alpha _{i,k}}\sum\nolimits_{l = 1}^{j - 1} {{P_l}}  + 1}}} \right)} ,}&{{\rm{IR}}}
\end{array}} \right..
\end{align}
where the term ${\alpha _{i,k}}\sum\nolimits_{l = 1}^{j - 1} {{P_l}}$ exists in the denominator  because the undetected messages (i.e., ${\bf x}_1,\cdots,{\bf x}_{j-1}$) are treated as noise while decoding ${\bf x}_j$.

In analogous to \eqref{eqn:out_user1}, the outage probability of user $i$ is expressed as
\begin{equation}\label{eqn:out_prob_def_mum2}
p_{i,K}^{out} = \Pr \left\{ {\bigcup\nolimits_{j = i}^M {{I_{i \to j,K}} < {R_j}} } \right\}.
\end{equation}
By using Lemma \ref{the:mul_events_diver}, the asymptotic behaviours of $\Pr \left\{ {{{I_{i \to j,K}} < {R_j}} } \right\}$ for $j=i,\cdots,M$ determine the diversity order of user $i$. More precisely, the diversity order of user $i$ is given by
\begin{equation}\label{eqn:diver_order_userim2}
d_{i }= \min\{d_{i\to j}:j\in [i,M]\},
\end{equation}
where $d_{i\to j}$ denotes the diversity order associated with $\Pr \left\{ {{{I_{i \to j,K}} < {R_j}} } \right\}$. By recognizing that the expression of ${I_{i \to j,K}} $ in \eqref{eqn:noma_harq_inf_i2m2} takes the similar form as \eqref{eqn:noma_harq_inf_i2} except for $i=j=1$, it is clear that the analytical results in Section \ref{sec:diver_order_user2} apply to the derivation of $d_{i\to j}$. Moreover, if $i=j=1$, ${I_{i \to j,K}} $ degenerates to the expression of accumulated mutual information under conventional HARQ schemes as \eqref{eqn:noma_harq_inf_1}. Accordingly, the conclusion of full time diversity drawn in \cite{shi2016optimal,shi2017asymptotic} is applicable here, that is, $d_{1\to 1}=K$. As a consequence, the expressions of $d_{i\to j}$ in different types of HARQ-aided NOMA systems are generalized as 
\begin{equation}\label{eqn:diversity_order_ij}
{d_{i \to j}} = \left\{ {\begin{array}{*{20}{c}}
{K{{\left[ {1 - \left\lfloor {\frac{{{2^{{R_j}}} - 1}}{{{c_j}}}} \right\rfloor } \right]}^ + },}&{\rm{I}}\\
{{{\left[ {K - \left\lfloor {\frac{{{2^{{R_j}}} - 1}}{{{c_j}}}} \right\rfloor } \right]}^ + },}&{{\rm{CC}}}\\
{{{\left[ {K - \left\lfloor {\frac{{{R_j}}}{{{{\log }_2}(1 + {c_j})}}} \right\rfloor } \right]}^ + },}&{{\rm{IR}}}
\end{array}} \right.,
\end{equation}
where ${c_j} = {{{P_j}}}/{{\sum\nolimits_{l = 1}^{j - 1} {{P_l}} }}$ for $j=2,\cdots,M$ and we stipulate ${c_1}=\infty$. Moreover, since ${d_{i \to j}}$ is independent of $i$, i.e., $d_{1\to j}=\cdots=d_{j\to j}$ for $j=[1,M]$, \eqref{eqn:diver_order_userim2} can be rewritten as
\begin{equation}\label{eqn:diver_order_userim2rew}
d_{i }= \min\{d_{i\to i},d_{i+1}\},\,i\in [1,M-1],
\end{equation}
and $d_M = d_{M\to M}$ and $d_1=K$. From \eqref{eqn:diver_order_userim2rew}, the relationship between users' diversity order follows as $d_1 = d_2 \le \cdots \le d_M$. Clearly, the diversity order of the HARQ-aided NOMA system is constrained by the users with worse channel quality.


\subsubsection{Power-Efficient Transmission Strategy}\label{sec:power_eff}
As stated in Section \ref{sec:sys_mod}, we assume that the superposition signals are kept sent in all HARQ rounds to simplify the transmission scheme. Needless to say, this assumption results in the waste of the valuable power resource on the successfully delivered message. To address this issue, we can develop an power-efficient transmission strategy, in which only the unsuccessfully recovered messages need to be resent. By doing so, the power of the message that is unnecessary to resent can be saved. Bearing this idea in mind, we proceed to determine the diversity order of the HARQ-aided NOMA system based on the power-efficient transmission strategy. To further exemplify, we still consider the two-user HARQ-aided NOMA system. The diversity order of user 1 is analyzed first. To this end, we begin with the outage probability of user 1. The occurrence of the outage at user 1 implies that the message ${\bf x}_1$ will never be recovered before ${\bf x}_2$. As elaborated in Section \ref{sec:sys_mod}, ${\bf x}_2$ can be treated as an apparent message to user 1 in subsequent transmissions after the recovery of ${\bf x}_2$. Therefore, this amounts to insisting on resending the superposition signals in all HARQ rounds, no matter whether ${\bf x}_2$ is recovered by both users. Accordingly, the outage probability of user 1 for the simple transmission strategy is still applicable to the novel power-efficient one. As a result, the diversity order $d_1^{\rm eff}$ of user 1 for the novel strategy will keep unchanged as $d_1^{\rm eff} = d_1 = d_2$.

In contrast, it becomes much more challenging to tackle the diversity order of user 2. In a similar manner, the outage probability of user 2 is first obtained. The decoding failure at user 2 means that the message ${\bf x}_1$ could be recovered before ${\bf x}_2$. For notational convenience, we define the successful decoding of user 1 after $\ell$ HARQ rounds as event $E_\ell$, and denote by $\bar E$ as the outage event at user 1. By using the law of total probability, the outage probability of user 2 can be written as
\begin{align}\label{eqn:out_2_eff}
p_{2,K}^{{\rm eff},out} =& \sum\nolimits_{\ell {{ = }}1}^{K - 1} {\Pr \left\{ { {{I_{2 \to 2,K,\ell }} < {R_2}} } \right\}{\Pr \left\{ {{E_\ell }} \right\}}}  \notag\\
&+ \Pr \left\{ { {{I_{2 \to 2,K}} < {R_2}} } \right\}{\Pr \left\{ \bar E\bigcup {{E_K}} \right\}},
\end{align}
where ${{I_{2 \to 2,K,\ell }}}$ is referred to as the mutual information accumulated by user $2$ for decoding ${\bf x}_2$ after $K$ HARQ rounds conditioned on $E_\ell$. By utilizing different types of HARQ schemes, ${{I_{2 \to 2,K,\ell }}}$ is explicitly given by \eqref{eqn:noma_harq_inf_i2_eff}, as shown at the top of the next page.
\begin{figure*}[!t]
\begin{align}\label{eqn:noma_harq_inf_i2_eff}
{I_{2 \to 2,K,\ell }} = \left\{ {\begin{array}{*{20}{c}}
{\max \left\{ {\left\{ {{{\log }_2}\left( {1 + \frac{{{\alpha _{2,k}}{P_2}}}{{{\alpha _{2,k}}{P_1} + 1}}} \right):k \in \left[ {1,\ell } \right]} \right\},\left\{ {{{\log }_2}\left( {1 + {\alpha _{2,k}}{P_2}} \right):k \in \left[ {\ell  + 1,K} \right]} \right\}} \right\},}&{\rm{I}}\\
{{{\log }_2}\left( {1 + \sum\nolimits_{k = 1}^\ell  {\frac{{{\alpha _{2,k}}{P_2}}}{{{\alpha _{2,k}}{P_1} + 1}}}  + \sum\nolimits_{k = \ell  + 1}^K {{\alpha _{2,k}}{P_2}} } \right),}&{{\rm{CC}}}\\
{\sum\nolimits_{k = 1}^\ell  {{{\log }_2}\left( {1 + \frac{{{\alpha _{2,k}}{P_2}}}{{{\alpha _{2,k}}{P_1} + 1}}} \right)}  + \sum\nolimits_{k = \ell  + 1}^K {{{\log }_2}\left( {1 + {\alpha _{2,k}}{P_2}} \right)} ,}&{{\rm{IR}}}
\end{array}} \right.
\end{align}
\hrulefill
\end{figure*}
In \eqref{eqn:noma_harq_inf_i2_eff}, the terms ${\alpha _{2,k}}{P_1}$ in the denominator disappear for $k \in \left[ {\ell  + 1,K} \right]$ because of the power-efficient transmissions. Note that ${\Pr \left\{ {{E_\ell }} \right\}}=p_{1,\ell-1}^{out}-p_{1,\ell}^{out}$, ${\Pr \left\{ {{\bar E }} \right\}}=p_{1,K}^{out}$ \cite{caire2001throughput}, and the events $\bar E $ and ${{E_K }} $ are mutually exclusive, \eqref{eqn:out_2_eff} is reduced to
\begin{align}\label{eqn:out_i2_eff_red}
p_{2,K}^{{\rm{eff}},out} =& \sum\nolimits_{\ell {\rm{ = }}1}^{K - 1} {\Pr \left\{ {{I_{2 \to 2,K,\ell }} < {R_2}} \right\}\left( {p_{1,\ell  - 1}^{out} - p_{1,\ell }^{out}} \right)}\notag\\
 & + p_{2,K}^{out}p_{1,K - 1}^{out}.
\end{align}
Clearly from \eqref{eqn:out_i2_eff_red}, the power-efficient strategy performs better than the simple strategy in terms of the outage probability of user 2, i.e., $p_{2,K}^{{\rm{eff}},out} \le p_{2,K}^{out}$, because $\Pr \left\{ {{I_{2 \to 2,K,\ell }} < {R_2}} \right\} \le \Pr \left\{ {{I_{2 \to 2,K}} < {R_2}} \right\}$. Furthermore, to obtain the diversity order $d_2^{\rm eff}$, it is mandatory to study the asymptotic behavior of $p_{2,K}^{{\rm{eff}},out}$ against $P_2$. Recalling that the asymptotic scaling laws of ${p_{1,\ell }^{out}}$ and $p_{2,K}^{out}$ as
\begin{equation}\label{eqn:asy1}
{p_{1,\ell }^{out}} = \Theta({P_1}^{-{d_{1,\ell }}}) = \Theta({P_2}^{-{d_{1,\ell }}}),
\end{equation}
\begin{equation}\label{eqn:asy2}
p_{2,K}^{out} = \Theta({P_2}^{-d_2}),
\end{equation}
where ${d_{1,\ell }}$ can be obtained from \eqref{eqn:diver_order_I}, \eqref{eqn:diversity_order_cc} and \eqref{eqn:diversity_order_ir} by simply replacing $K$ with $\ell$. 
In addition, Appendix \ref{app:div_eff} proves
\begin{equation}\label{eqn:asy3}
\Pr \left\{ {{I_{2 \to 2,K,\ell }} < {R_2}} \right\} = \Theta({P_2}^{-d_{2,\ell}}),
\end{equation}
where
\begin{equation}\label{eqn:d2_ell}
d_{2,\ell} = \left\{ {\begin{array}{*{20}{c}}
{\ell {{\left[ {1 - \left\lfloor {\frac{{{2^{{R_2}}} - 1}}{c}} \right\rfloor } \right]}^ + } + K - \ell ,}&{\rm{I}}\\
{{{\left[ {\ell  - \left\lfloor {\frac{{{2^{{R_2}}} - 1}}{c}} \right\rfloor } \right]}^ + } + K - \ell ,}&{{\rm{CC}}}\\
{{{\left[ {\ell  - \left\lfloor {\frac{{{R_2}}}{{{{\log }_2}(1 + c)}}} \right\rfloor } \right]}^ + } + K - \ell ,}&{{\rm{IR}}}
\end{array}} \right.,
\end{equation}
By plugging \eqref{eqn:asy1}-\eqref{eqn:asy3} into \eqref{eqn:out_i2_eff_red}, we arrive at
\begin{align}\label{eqn:outeff_d2}
&p_{2,K}^{{\rm{eff}},out} 
=\sum\nolimits_{\ell {{ = }}1}^{K - 1} {\Theta \left( {{P_2}^{ - {d_{1,\ell  - 1}} - {d_{2,\ell }}}} \right)}  + \Theta \left( {{P_2}^{ - {d_{1,K - 1}} - {d_2}}} \right)\notag\\
&=\Theta \left( {{P_2}^{ - \min \left\{ {{d_{1,\ell  - 1}} + {d_{2,\ell }}:\ell  \in \left[ {1,K} \right]} \right\}}} \right)=\Theta \left( {{P_2}^{ - d_2}} \right),
\end{align}
Accordingly, the diversity order $d_2^{\rm eff}$ of user 2 is $d_2^{\rm eff}=d_2$.
In a nutshell, the power-efficient transmission strategy can decrease the outage probability of the farther user, but the corresponding diversity order remains the same as under the simple strategy.

\subsubsection{Imperfect CSI}\label{sec:power_eff}
To account for the impact of imperfect CSI, the minimum mean square error (MMSE) channel estimation error model in \cite{yang2015performance} is assumed herein. We denote by $\hat h_{i,k}$ the estimate of the channel coefficient $ h_{i,k}$ between the source and user $i$ in the $k$-th HARQ round. Accordingly, the channel estimation error $\epsilon_{i,k}$ follows as $\epsilon_{i,k} = h_{i,k} - \hat h_{i,k}$, which is complex Gassuian distributed with mean $0$ and variance $\sigma_i^2$, i.e., $\epsilon_{i,k} \sim {\cal CN}(0,\sigma_i^2)$. By recalling the Rayleigh fading channels together with $\alpha_{i,k}= |h_{i,k}|^2$, $\hat h_{i,k}$ is also a complex Gaussian distributed with mean $0$ and variance $\bar \alpha_i - \sigma_i^2$, i.e., $\hat h_{i,k} \sim {\cal CN}(0,\bar \alpha_i - \sigma_i^2)$. Accordingly, the received signal at user $i$ in the $k$-th HARQ round is expressed as
\begin{equation}\label{eqn:imper_csi}
  {{\bf y}_{i,k}} = \hat h_{i,k} \sum\nolimits_{i=1}^M {\bf x}_i + \underbrace{\epsilon_{i,k} \sum\nolimits _{i=1}^M {\bf x}_i + {\bf w}_{i,k}}_{\triangleq\tilde{\bf w}_{i,k}}, 
\end{equation}
where ${\bf w}_{i,k}$ denotes AWGN with unity variance. In general the term ${\tilde{\bf w}_{i,k}}$ in \eqref{eqn:imper_csi} can be treated as Gaussian noise with variance $1+\sigma_i^2\sum\nolimits_{i=1}^MP_i$. As a result, \eqref{eqn:imper_csi} collapses to the case of perfect CSI. Hence, the proposed approach is still applicable to the case of imperfect CSI. Besides, the resultant diversity orders keep unchanged because of the invariable ratios between transmit powers allocated to different users.


\section{Numerical Results}\label{sec:num}
In this section, the numerical results are presented for verifications. Unless otherwise indicated, we take the two-user HARQ-aided NOMA system with $K=4$, $R_1=R_2=1$bps/Hz and $\bar \alpha_1=2 \bar\alpha_2=2$ as examples. The simulation values of the outage probability presented in Figs. \ref{fig:cc_ver}, \ref{fig:ir_ver} and \ref{fig:userm4} are computed by carrying out $10^7$ Monte Carlo runs. 

Fig. \ref {fig:I_ver} shows the comparison between the exact and simulation results in terms of the outage probability of user 2 in Type I HARQ-aided NOMA system. This figure demonstrates the excellent match between the exact and simulation results. Clearly, the outage event occurs if $c \le 2^{R_2}-1=1$ according to \eqref{eqn:diver_order_I}, which can be justified by Fig. \ref{fig:I_ver}. Otherwise, if $c> 1$, the diversity order of user 2 is equal to $K=4$, which can be observed from Fig. \ref {fig:I_ver}. Besides, it is not surprising that the outage probability decreases with the increase of $c$, because the transmission power allocated to user 2 rises.
\begin{figure}
  \centering
  \includegraphics[width=2in]{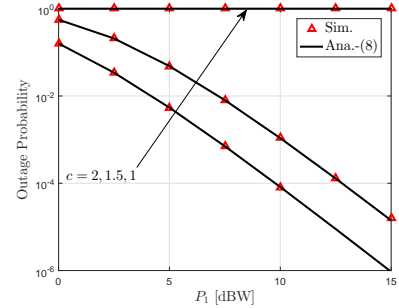}
  \caption{The outage probability of user 2 against $P_1$ in Type I HARQ-aided NOMA system.}\label{fig:I_ver}
\end{figure}

In Fig. \ref{fig:cc_ver}, the outage probability of user 2 is plotted versus $P_1$ in HARQ-CC-aided NOMA system, where the lower and upper bounds of the outage probability are obtained by using \eqref{eqn:g_red_fin}. It is noteworthy that the value of $\Delta$ in \eqref{eqn:g_red_fin} is set in each recursive step as $\Delta = (c-{\rm mod}(\gamma,c))/2$, where ${\rm mod}(\cdot,\cdot)$ is the operator of modulo. It is easily seen that the simulation results lie in between the corresponding lower and upper bounds. In addition, it is observed that the curves of the simulation results become parallel to those of lower and upper bounds in high SNR regime. This confirms the validity of the analysis of diversity order. Moreover, it is seen from Fig. \ref{fig:cc_ver} that the outage curve declines faster as $c$ increases, because \eqref{eqn:diversity_order_cc} indicates that the diversity orders associated with $c=0.4$, $c=0.8$ and $c=1.2$ are $2$, $3$ and $4$, respectively.
\begin{figure}
  \centering
  \includegraphics[width=2in]{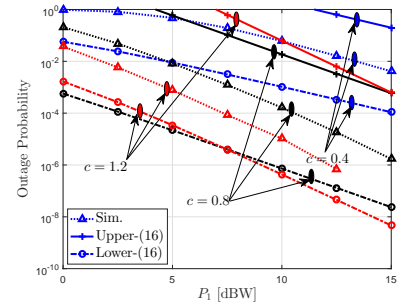}
  \caption{The outage probability of user 2 against $P_1$ in HARQ-CC-aided NOMA system.}\label{fig:cc_ver}
\end{figure}

Fig. \ref{fig:ir_ver} depicts the outage probability of user 2 against $P_1$ in HARQ-IR-aided NOMA system. The lower and upper bounds of the outage probability are obtained by utilizing \eqref{eqn:psi_lower_upper}. It is worth noting that the value of $\Delta$ in \eqref{eqn:psi_lower_upper} is set as $\Delta = (1+c-\exp({\rm mod}(\ln\gamma,\ln(1+c))))/2$. It is readily found that the simulation results are lower and upper bounded in Fig. \ref{fig:ir_ver}. Moreover, similarly to Fig. \ref{fig:cc_ver}, the numerical results with regard to the diversity order can be observed in Fig. \ref{fig:ir_ver}.
\begin{figure}
  \centering
  \includegraphics[width=2in]{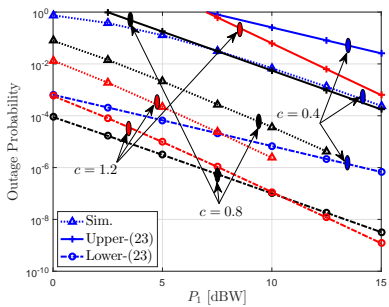}
  \caption{The outage probability of user 2 against $P_1$ in HARQ-IR-aided NOMA system.}\label{fig:ir_ver}
\end{figure}

Fig. \ref{fig:ord} illustrates the relationship between the diversity order and the transmission rate. It is observed that the diversity order is a decreasing step function of $R_2$, which is consistent with our analysis. In accordance with Fig. \ref{fig:ord}, the three types of HARQ-aided NOMA schemes are capable of achieving full time diversity (i.e., $d_2=K$) under the same condition, i.e., $R_2 < \log_2(1+c)$. Besides, by fixing the values of $c$ and $R_2$, HARQ-IR-aided NOMA performs the best among the three types of HARQ-aided NOMA in terms of diversity order. Besides, Type I HARQ-aided NOMA exhibits the lowest diversity order. It should be mentioned that higher diversity order is achieved by HARQ-IR in sacrifice of the computation complexity. In addition, the diversity order can also be improved by increasing $c$.

\begin{figure}
  \centering
  \includegraphics[width=2in]{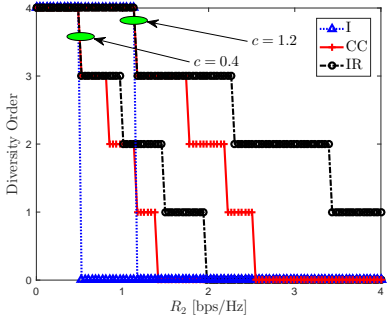}
  \caption{The diversity order of user 2 against $R_2$.}\label{fig:ord}
\end{figure}

In Fig. \ref{fig:userm4}, the outage performance of four-user HARQ-aided NOMA systems is investigated. The curves of the outage probabilities for four users are obtained by conducting Monte Carlo simulations. By following \eqref{eqn:diver_order_userim2rew}, all the users except user 4 in the Type I HARQ-aided NOMA system obtain zero diversity order, because the diversity orders of users 1 to 3 are restricted by $d_{3\to 3} = K[1-\lfloor (2^{R_3}-1)/c_3\rfloor]^+=0$. This result can be observed in Fig. \ref{fig:userm4}, where the outage event takes place at users 1 to 3. Moreover, full diversity order is achievable by user 4 in any type of HARQ-aided NOMA system, because $R_4 < \log_2(1+c_4)$ is satisfied. For instance, by taking the HARQ-CC-aided NOMA scheme in Fig. \ref{fig:userm4} as an example, the outage probabilities of user 4 are around $2\times 10^{-2}$ and $2\times 10^{-5}$ by fixing $P_1=0$dBW and $10$dBW, respectively. According to the definition of ${\tilde d_i}$ in Footnote \ref{ft:def_d}, the diversity order is about $\tilde d_i = 10 (\log_{10}2\times 10^{-2} - \log_{10}2\times 10^{-5})/(10-0)=3$. This further corroborates the correctness of our analysis. In addition, the diversity orders of users 1, 2 and 3 in the HARQ-CC-aided NOMA system are all equal to $1$, which can be justified by Fig. \ref{fig:userm4}. Furthermore, the diversity orders of users 1, 2 and 3 in the HARQ-IR-aided NOMA system are all equal to $2$, which are in perfect agreement with the simulation results in Fig. \ref{fig:userm4}. Unsurprisingly, the HARQ-IR-aided NOMA system outperforms the other two types of HARQ-aided NOMA systems owing to the exploitation of additional coding gain. 

\begin{figure}
  \centering
  \includegraphics[width=2in]{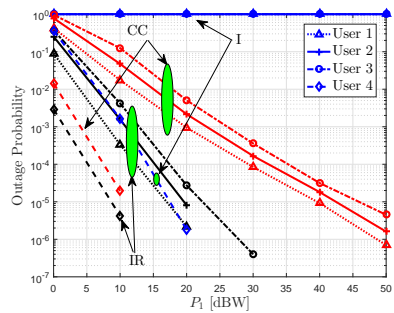}
  \caption{The outage probability of the HARQ-aided NOMA system with parameters $M=4$, $K=3$, $c_2=2$, $c_3 =1.4$, $c_4 =4$, $\bar\alpha_1=2\bar\alpha_2=4\bar\alpha_3=6\bar\alpha_4=2$ and $R_1=R_2=R_3=R_4=2$bps/Hz.}\label{fig:userm4}
\end{figure}
\section{Conclusions}\label{sec:con}
The diversity order is frequently used to capture the asymptotic behaviour of the outage probability under high SNR. This paper has thoroughly investigated the diversity orders of three types of HARQ-aided NOMA systems, i.e., Type I HARQ, HARQ-CC and HARQ-IR. The diversity orders of three different HARQ-aided downlink NOMA systems have been derived in closed-form, wherein the integration domain partition trick has been proposed to provide outage bounds especially for HARQ-CC and HARQ-IR-aided NOMA systems. The analytical results have verified that the diversity order is a decreasing step function of the transmission rate given the ratio between the transmit powers allocated to the two users. It has been proved that HARQ-IR-aided NOMA systems achieve the largest diversity order, followed by HARQ-CC-aided and then Type I HARQ-aided NOMA systems. Besides, users' diversity orders comply with a descending order according to their respective average channel gains. The analytical results are also applicable to the case with an arbitrary number of users. Furthermore, we have expanded some discussions concerning the impacts of both power-efficient transmission strategy and imperfect CSI.

More importantly, the question raised in Section I can now be answered. The answer is in the negative, more precisely, full time diversity can be achieved by HARQ-aided NOMA systems only under very specific cases. For example, in two-user HARQ-aided NOMA systems, the transmission rate should be set below a threshold that depends on the HARQ type and the ratio between the transmit powers allocated for the weak and strong users.

\appendices
\section{Proof of \eqref{eqn:g_red_fin}}\label{app:harq_cc_boundphi}
\eqref{eqn:cdf_gamma_sum2} can be further expanded as \eqref{eqn:g_redefine} at the top of the next page, where ${z_{1:k}}$ represents the sequence $z_1,\cdots,z_k$.
\begin{figure*}[!t]
\begin{align}\label{eqn:g_redefine}
{\phi_K}\left( \gamma  \right)& = \int\nolimits_0^c {{ \cdots \int\nolimits_0^c \int\nolimits_{\min \left( {\max \left( {Kc - \gamma  - \sum\nolimits_{k = 1}^{K - 1} {{z_k}} ,0} \right),c} \right)}^c {{e^{ - \frac{c}{{\bar \alpha_2 {P_1}}}\sum\nolimits_{k = 1}^K {\frac{1}{{{z_k}}}} }}\prod\nolimits_{k = 1}^K {\frac{1}{{{z_k}^2}}} d{z_1}d{z_2} \cdots d{z_K}} } }\notag\\
 &= \int\nolimits_{\scriptstyle{z_{1:K - 1}} \in \left[ {0,c} \right],\hfill\atop
\scriptstyle\left( {K - 1} \right)c - \gamma  \le \sum\nolimits_{k = 1}^{K - 1} {{z_k}}  < Kc - \gamma \hfill} {{e^{ - \frac{c}{{{{\bar \alpha }_2}{P_1}}}\sum\nolimits_{k = 1}^{K - 1} {\frac{1}{{{z_k}}}} }}\prod\nolimits_{k = 1}^{K - 1} {\frac{1}{{{z_k}^2}}} d{z_1} \cdots d{z_{K - 1}}} \int\nolimits_{Kc - \gamma  - \sum\nolimits_{k = 1}^{K - 1} {{z_k}} }^c {{e^{ - \frac{c}{{{{\bar \alpha }_2}{P_1}}}\frac{1}{{{z_K}}}}}\frac{1}{{{z_K}^2}}d{z_K}}  \notag \\
 &\quad +  \int\nolimits_{\scriptstyle{z_{1:K - 1}} \in \left[ {0,c} \right],\hfill\atop
\scriptstyle\sum\nolimits_{k = 1}^{K - 1} {{z_k}}  \ge Kc - \gamma \hfill} {{{e^{ - \frac{c}{{\bar \alpha_2 {P_1}}}\sum\nolimits_{k = 1}^{K - 1} {\frac{1}{{{z_k}}}} }}}}  \prod\nolimits_{k = 1}^{K - 1} {\frac{1}{{{z_k}^2}}} d{z_1} \cdots d{z_{K - 1}}\int\nolimits_0^c {{e^{ - \frac{c}{{\bar \alpha_2 {P_1}}}\frac{1}{{{z_K}}}}}\frac{1}{{{z_K}^2}}d{z_K}},
\end{align}
\end{figure*}
With the definition of ${\phi _K}\left( \gamma  \right)$, \eqref{eqn:g_redefine} can be further simplified as \eqref{eqn:g_redefine1}, shown at the top of the next page.
\begin{figure*}[!t]
\begin{align}\label{eqn:g_redefine1}
{\phi _K}\left( \gamma  \right) =& \frac{{{{\bar \alpha }_2}{P_1}}}{c}\int\nolimits_{\scriptstyle{z_{1:K - 1}} \in \left[ {0,c} \right],\hfill\atop
\scriptstyle\left( {K - 1} \right)c - \gamma  \le \sum\nolimits_{k = 1}^{K - 1} {{z_k}}  < Kc - \gamma \hfill} {{e^{ - \frac{c}{{{{\bar \alpha }_2}{P_1}}}\sum\nolimits_{k = 1}^{K - 1} {\frac{1}{{{z_k}}}} }}\prod\nolimits_{k = 1}^{K - 1} {\frac{1}{{{z_k}^2}}} \left( {{e^{ - \frac{1}{{{{\bar \alpha }_2}{P_1}}}}} - {e^{ - \frac{c}{{{{\bar \alpha }_2}{P_1}\left( {Kc - \gamma  - \sum\nolimits_{k = 1}^{K - 1} {{z_k}} } \right)}}}}} \right)d{z_1} \cdots d{z_{K - 1}}} \notag\\
 &+ \frac{{{{\bar \alpha }_2}{P_1}}}{c}{e^{ - \frac{1}{{{{\bar \alpha }_2}{P_1}}}}}{\phi _{K - 1}}(\gamma  - c).
\end{align}
\end{figure*}
The lower bound of ${\phi_K}\left( \gamma  \right)$ in \eqref{eqn:g_red_fin} is got because the first term in \eqref{eqn:g_redefine1} is non-negative. In order to obtain the upper bound of ${\phi_K}\left( \gamma  \right)$ in \eqref{eqn:g_redefine1}, ${\phi_K}\left( \gamma  \right)$ can be expressed as \eqref{eqn:g_red_approx}, shown at the top of the next page, where $0 < \Delta \le c$.
\begin{figure*}[!t]
\begin{align}\label{eqn:g_red_approx}
{\phi _K}\left( \gamma  \right) &= \frac{{{{\bar \alpha }_2}{P_1}}}{c}\int\nolimits_{\scriptstyle{z_{1:K - 1}} \in \left[ {0,c} \right],\hfill\atop
\scriptstyle\left( {K - 1} \right)c - \gamma  \le \sum\nolimits_{k = 1}^{K - 1} {{z_k}}  < Kc - \gamma  - \Delta \hfill} {{e^{ - \frac{c}{{{{\bar \alpha }_2}{P_1}}}\sum\nolimits_{k = 1}^{K - 1} {\frac{1}{{{z_k}}}} }}} \left( {{e^{ - \frac{1}{{{{\bar \alpha }_2}{P_1}}}}} - {e^{ - \frac{c}{{{{\bar \alpha }_2}{P_1}\left( {Kc - \gamma  - \sum\nolimits_{k = 1}^{K - 1} {{z_k}} } \right)}}}}} \right)\prod\nolimits_{k = 1}^{K - 1} {\frac{1}{{{z_k}^2}}} d{z_1} \cdots d{z_{K - 1}} \notag\\
 &+ \frac{{{{\bar \alpha }_2}{P_1}}}{c}\int\nolimits_{\scriptstyle{z_{1:K - 1}} \in \left[ {0,c} \right],\hfill\atop
\scriptstyle Kc - \gamma  - \Delta  \le \sum\nolimits_{k = 1}^{K - 1} {{z_k}}  < Kc - \gamma \hfill} {{e^{ - \frac{c}{{{{\bar \alpha }_2}{P_1}}}\sum\nolimits_{k = 1}^{K - 1} {\frac{1}{{{z_k}}}} }}} \left( {{e^{ - \frac{1}{{{{\bar \alpha }_2}{P_1}}}}} - {e^{ - \frac{c}{{{{\bar \alpha }_2}{P_1}\left( {Kc - \gamma  - \sum\nolimits_{k = 1}^{K - 1} {{z_k}} } \right)}}}}} \right)\prod\nolimits_{k = 1}^{K - 1} {\frac{1}{{{z_k}^2}}} d{z_1} \cdots d{z_{K - 1}} \notag\\
& + \frac{{{{\bar \alpha }_2}{P_1}}}{c}{e^{ - \frac{1}{{{{\bar \alpha }_2}{P_1}}}}}{\phi _{K - 1}}(\gamma  - c),
\end{align}
\end{figure*}
Therefore, ${\phi _K}\left( \gamma  \right)$ is upper bounded as \eqref{eqn:g_red_approx1}, as shown at the top of this page.
\begin{figure*}[!t]
\begin{align}\label{eqn:g_red_approx1}
&{\phi _K}\left( \gamma  \right) \le \frac{{{{\bar \alpha }_2}{P_1}}}{c}\left( {{e^{ - \frac{1}{{{{\bar \alpha }_2}{P_1}}}}} - {e^{ - \frac{1}{{{{\bar \alpha }_2}{P_1}}}\frac{c}{\Delta }}}} \right)\underbrace{\int\nolimits_{\scriptstyle{z_{1:K - 1}} \in \left[ {0,c} \right]\hfill\atop
\scriptstyle\left( {K - 1} \right)c - \gamma  \le \sum\nolimits_{k = 1}^{K - 1} {{z_k}}  < Kc - \gamma  - \Delta \hfill} {{e^{ - \frac{c}{{{{\bar \alpha }_2}{P_1}}}\sum\nolimits_{k = 1}^{K - 1} {\frac{1}{{{z_k}}}} }}\prod\nolimits_{k = 1}^{K - 1} {\frac{1}{{{z_k}^2}}} d{z_1} \cdots d{z_{K - 1}}}}_{\le {{\phi _{K - 1}}(\gamma)}} \notag\\
& + \frac{{{{\bar \alpha }_2}{P_1}}}{c}{e^{ - \frac{1}{{{{\bar \alpha }_2}{P_1}}}}} \underbrace{\int\nolimits_{\scriptstyle{z_{1:K - 1}} \in \left[ {0,c} \right],\hfill\atop
\scriptstyle Kc - \gamma  - \Delta  \le \sum\nolimits_{k = 1}^{K - 1} {{z_k}}  < Kc - \gamma \hfill} {{e^{ - \frac{c}{{{{\bar \alpha }_2}{P_1}}}\sum\nolimits_{k = 1}^{K - 1} {\frac{1}{{{z_k}}}} }}\prod\nolimits_{k = 1}^{K - 1} {\frac{1}{{{z_k}^2}}} d{z_1} \cdots d{z_{K - 1}}} }_{={\phi _{K - 1}}(\gamma  - c + \Delta ) - {\phi _{K - 1}}(\gamma  - c)} + \frac{{{{\bar \alpha }_2}{P_1}}}{c}{e^{ - \frac{1}{{{{\bar \alpha }_2}{P_1}}}}}{\phi _{K - 1}}(\gamma  - c).
\end{align}
\hrulefill
\end{figure*}
By using the inequality $1 - {e^{ - x}} \le x$, ${\phi _K}\left( \gamma  \right)$ is further upper bounded as \eqref{eqn:g_red_fin}.

\section{Proof of Lemma \ref{the:cc_phi}}\label{app:proof_phi_gener}
\subsection{Case 1: $\gamma = 0$}
If $\gamma = 0$, it is easily proved from the definition of ${\phi _K}\left( \gamma  \right)$ in \eqref{eqn:cdf_gamma_sum2} that ${\phi _K}\left( \gamma  \right)=0$ because $u\left( {\gamma  + \sum\nolimits_{k = 1}^K {{z_k}}  - Kc} \right)=0$.
\subsection{Case 2: $0 < \gamma < c$}
We next consider the case of $0<\gamma < c$. Notice $0 \le z_k \le c$, ${\sum\nolimits_{k = 1}^K {{z_k}} \ge Kc - \gamma}$ is found to be a subset of $\left\{ {\left( {{z_1},\cdots,{z_K}} \right):{z_k} \ge c - \gamma} \right\}$. Thus, ${\phi _K}\left( \gamma  \right)$ is upper bounded according to its definition as
\begin{align}\label{eqn:phi_K_lower_cc}
&{\phi _K}\left( \gamma  \right) \le \int\nolimits_{c - \gamma }^c { \cdots \int\nolimits_{c - \gamma }^c {{e^{ - \frac{c}{{\bar \alpha_2 {P_1}}}\sum\nolimits_{k = 1}^K {\frac{1}{{{z_k}}}} }}\prod\limits_{k = 1}^K {\frac{1}{{{z_k}^2}}} d{z_1} \cdots d{z_K}} } =\notag\\
 & {\left( {\frac{{\bar \alpha_2 {P_1}}}{c}\left( {{e^{ - \frac{1}{{\bar \alpha_2 {P_1}}}}} - {e^{ - \frac{c}{{\bar \alpha_2 {P_1}}({{c - \gamma }})}}} }\right)} \right)^K}= {e^{ - K{\varepsilon _1}}}{\left( {\frac{1}{{c - \gamma }} - \frac{1}{c}} \right)^K},
\end{align}
where the last inequality holds by using lagrange's mean value theorem and ${1}/{\left({\bar \alpha_2 {P_1}}\right)} \le \varepsilon_1  \le {c}/{\left({\bar \alpha_2 {P_1}}{({c - \gamma })}\right)}$. By applying $\varepsilon_1 \ge {1}/{\left({\bar \alpha_2 {P_1}}\right)}$ to \eqref{eqn:phi_K_lower_cc}, the upper bound of \eqref{eqn:phi_K_cc_bounds_c1} follows.

Similarly, since $\left\{ {\left( {{z_1},\cdots,{z_K}} \right):{z_k} \ge c - {\gamma}/{K}} \right\}$ is a subset of ${\sum\nolimits_{k = 1}^K {{z_k}} \ge Kc - \gamma}$, ${\phi _K}\left( \gamma  \right)$ is lower bounded as
\begin{align}\label{eqn:phi_K_upper}
&{\phi _K}\left( \gamma  \right) \ge \int\nolimits_{c - {\gamma }/{K}}^c { \cdots \int\nolimits_{c - {\gamma }/{K}}^c {{e^{ - \frac{c}{{\bar \alpha_2 {P_1}}}\sum\nolimits_{k = 1}^K {\frac{1}{{{z_k}}}} }}\prod\nolimits_{k = 1}^K {\frac{1}{{{z_k}^2}}} d{z_1} \cdots d{z_K}} } \notag\\
 &={\left( {\frac{{{{\bar \alpha }_2}{P_1}}}{c}\left( {{e^{ - \frac{1}{{{{\bar \alpha }_2}{P_1}}}}} - {e^{ - \frac{c}{{{{\bar \alpha }_2}{P_1}\left( {c - \frac{\gamma }{K}} \right)}}}}} \right)} \right)^K}= {e^{ - K{\varepsilon _2}}}{\left( {\frac{1}{{c - \frac{\gamma }{K}}} - \frac{1}{c}} \right)^K},
\end{align}
where the last inequality holds by using lagrange's mean value theorem and ${1}/{\left({\bar \alpha_2 {P_1}}\right)} \le \varepsilon_2  \le {c}/{\left({\bar \alpha_2 {P_1}}{({c - \gamma /K})}\right)}$. Notice $\varepsilon_2  \le {c}/{\left({\bar \alpha_2 {P_1}}{({c - \gamma /K})}\right)}$, the lower bound of ${\phi _K}\left( \gamma  \right)$ is obtained as \eqref{eqn:phi_K_cc_bounds_c1}.

According to the definition of ${\phi _K}\left( \gamma  \right)$ in \eqref{eqn:cdf_gamma_sum2}, ${\phi _K}\left( \gamma  \right)$ is an increasing function of $P_1$ together with the upper bound of ${\phi _K}\left( \gamma  \right)$ in \eqref{eqn:phi_K_cc_bounds_c1}, ${\phi _K}\left( \gamma  \right)$ is convergent as $P_1$ tends to infinity. Moreover, the limit of ${\phi _K}\left( \gamma  \right)$ as $P_1 \to \infty$ is non-zero and bounded as
\begin{equation}\label{eqn:phi_K_boudns}
{\left( {\frac{1}{{c - {\gamma }/{K}}} - \frac{1}{c}} \right)^K} \le \mathop {\lim }\limits_{{P_1} \to \infty } {\phi _K}\left( \gamma  \right) \le {\left( {\frac{1}{{c - \gamma }} - \frac{1}{c}} \right)^K}.
\end{equation}
Hence, $\mathop {\lim }\nolimits_{{P_1} \to \infty } {\phi _K}\left( \gamma  \right) = \rm const$ is thus proved for $0 < \gamma < c$.
\subsection{Case 3: $\gamma \ge Kc$}
If $\gamma \ge Kc$, we have $u\left( {\gamma  + \sum\nolimits_{k = 1}^K {{z_k}}  - Kc} \right)=1$. By putting this result into \eqref{eqn:cdf_gamma_sum2}, ${\phi _K}\left( \gamma  \right)$ is thus derived as
\begin{align}\label{eqn:phi_K_lower_ccc3}
{\phi _K}\left( \gamma  \right) &= \int\nolimits_{0}^c { \cdots \int\nolimits_{0}^c {{e^{ - \frac{c}{{\bar \alpha_2 {P_1}}}\sum\nolimits_{k = 1}^K {\frac{1}{{{z_k}}}} }}\prod\limits_{k = 1}^K {\frac{1}{{{z_k}^2}}} d{z_1} \cdots d{z_K}} } \notag\\
 &={{\left( {\frac{{\bar \alpha_2 }{P_1}}{c}} \right)}^K}{e^{ - \frac{K}{{\bar \alpha_2 {P_1}}}}}=\Theta({P_1}^K).
\end{align}
\subsection{Case 4: $kc \le \gamma  < (k+1)c$ for $k\in [1,K-1]$}\label{sec:case4_cc}
We first consider the case of $kc < \gamma  < (k+1)c$. By successively using the lower and upper bounds in \eqref{eqn:g_red_fin} along with the results in Cases 2 and 3, the squeeze theorem suggests ${\phi _K}\left( \gamma  \right) = \Theta({P_1}^k)$. It is worthwhile to mention that $\Delta$ should be properly chosen to obtain a tight upper bound for ${\phi _K}\left( \gamma  \right) $. Herein, $\Delta$ could be set as $\Delta = (c-{\rm mod}(\gamma,c))/2$ in each recursive step, where ${\rm mod}(\cdot,\cdot)$ denotes the modulo operation. Moreover, if $\gamma = kc$, we can no longer rely solely on the same approach to derive the asymptotic expression of ${\phi _K}\left( \gamma  \right)$. Although repeatedly using the upper bound in \eqref{eqn:g_red_fin} still yields ${\phi _K}\left( kc  \right) \le \Theta({P_1}^k)$. Whereas, the term $\phi_{K-k}(0)=0$ occurs if we apply the lower bound in \eqref{eqn:g_red_fin} recursively, which leads to a useless result. Instead, we turn to study the asymptotic behaviour of ${\phi _\kappa}\left( c  \right)$, where $\kappa=K-k+1$. ${{\phi _\kappa}\left( c \right)}$ can be rewritten by using the definition \eqref{eqn:cdf_gamma_sum2} as
\begin{align}\label{eqn:phi_K_c}
&{\phi _\kappa}\left( c \right) = \int_0^c {{e^{ - \frac{c}{{{P_1}{{\bar \alpha }_2}}}\frac{1}{{{z_\kappa}}}}}\frac{1}{{{z_\kappa}^2}}d{z_\kappa}} \int_0^c { \cdots \int_0^c {{e^{ - \frac{c}{{{P_1}{{\bar \alpha }_2}}}\sum\nolimits_{k = 1}^{\kappa - 1} {\frac{1}{{{z_k}}}} }} } }\notag\\
&\times \prod\nolimits_{k = 1}^{\kappa - 1} {\frac{1}{{{z_k}^2}}}u\left( {{z_\kappa} + \sum\nolimits_{k = 1}^{\kappa - 1} {{z_k}}  - \left( {\kappa - 1} \right)c} \right)d{z_1} \cdots d{z_{\kappa - 1}}\notag\\
&= \int_0^c {{e^{ - \frac{c}{{{P_1}{{\bar \alpha }_2}}}\frac{1}{{{z_\kappa}}}}}\frac{{{\phi _{\kappa - 1}}\left( {{z_\kappa}} \right)}}{{{z_\kappa}^2}}d{z_\kappa}}.
\end{align}
Thanks to the increasing monotonicity of ${\phi _\kappa}\left( z  \right)$ with respect to $z$, \eqref{eqn:phi_K_c} can be further expressed by using the intermediate value theorem as
\begin{align}\label{eqn:phi_K_c_fin}
{\phi _\kappa }\left( c \right) &= {\phi _{\kappa - 1}}\left( \xi_{P_1}  \right)\int_0^c {\underbrace{{e^{ - \frac{c}{{{P_1}{{\bar \alpha }_2}}}\frac{1}{{{z_\kappa}}}}}\frac{1}{{{z_\kappa}^2}}}_{g(z_k)}d{z_\kappa}} \notag\\
&= \frac{{P_1}{{\bar \alpha }_2}}{c}{e^{ - \frac{1}{{{P_1}{{\bar \alpha }_2}}}}}{\phi _{\kappa - 1}}\left( \xi_{P_1}  \right)=\Theta(P_1),
\end{align}
where $\xi_{P_1}  \in \left( {0,c} \right)$ and the last equality holds by using the following theorem.
\begin{theorem}\label{the:inter_med}
As $P_1$ approaches to infinity, we have $\mathop {\lim }\nolimits_{{P_1} \to \infty } {\phi _{\kappa - 1}}\left( \xi_{P_1}  \right)={\rm const} \ne 0$ and $\mathop {\lim }\nolimits_{{P_1} \to \infty } \xi_{P_1} \ne 0$.
\end{theorem}
\begin{proof}
Recalling $\xi_{P_1}  \in \left( {0,c} \right)$, the result of Case 2 indicates the limit of ${\phi _{\kappa - 1}}\left( \xi_{P_1}  \right) $ exists as $P_1 \to \infty$. We proceed by contradiction. Suppose that $\mathop {\lim }\nolimits_{{P_1} \to \infty } {\phi _{\kappa - 1}}\left( \xi_{P_1}  \right) = 0$, $\int_0^c {{\phi _{\kappa  - 1}}\left( {{z_\kappa }} \right)g\left( {{z_\kappa }} \right)d{z_\kappa }}  = {\phi _{\kappa  - 1}}\left( \xi_{P_1}  \right)\int_0^c {g\left( {{z_\kappa }} \right)d{z_\kappa }} $ indicates
\begin{align}\label{eqn:phi_K}
0 &\coloneqq \mathop {\lim }\limits_{{P_1} \to \infty } \int_0^c {\left( {{\phi _{\kappa  - 1}}\left( {{z_\kappa }} \right) - {\phi _{\kappa  - 1}}\left( \xi_{P_1}  \right)} \right)g\left( {{z_\kappa }} \right)d{z_\kappa }} \notag\\
&= \int_0^c {\mathop {\lim }\limits_{{P_1} \to \infty } {\phi _{\kappa  - 1}}\left( {{z_\kappa }} \right)\mathop {\lim }\limits_{{P_1} \to \infty } g\left( {{z_\kappa }} \right)d{z_\kappa }}.
\end{align}
However, since ${\mathop {\lim }\nolimits_{{P_1} \to \infty } {\phi _{\kappa  - 1}}\left( {{z_\kappa }} \right)} > 0$ and ${\mathop {\lim }\nolimits_{{P_1} \to \infty } g\left( {{z_\kappa }} \right)} > 0$ for $0<z_\kappa\le c$, it follows that $\int_0^c {\mathop {\lim }\nolimits_{{P_1} \to \infty } {\phi _{\kappa  - 1}}\left( {{z_\kappa }} \right)\mathop {\lim }\nolimits_{{P_1} \to \infty } g\left( {{z_\kappa }} \right)d{z_\kappa }}>0$, which contradicts \eqref{eqn:phi_K}. Thus, $\mathop {\lim }\nolimits_{{P_1} \to \infty } {\phi _{\kappa - 1}}\left( \xi_{P_1}  \right) \ne 0$. By using the result of Case 2, $\mathop {\lim }\nolimits_{{P_1} \to \infty } \xi_{P_1} \ne 0$ is proved.
\end{proof}

By combining the recursive approach with \eqref{eqn:phi_K_c_fin}, ${\phi _K}\left( kc  \right) \ge \Theta({P_1}^k)$ follows. Together with the obtained upper bound ${\phi _K}\left( kc  \right) \le \Theta({P_1}^k)$, the squeeze theorem implies ${\phi _K}\left( kc  \right) = \Theta({P_1}^k)$. We thus complete the proof.
\section{Proof of (\ref{eqn:cdf_gamma_sum_ir_fin})}\label{app:proof_ir_out}
The Mellin transform of the probability density function (PDF) of $\tilde \gamma$ can be derived as
\begin{align}\label{eqn:gam_tilde_mt}
&\mathbb E\left\{ {{{\tilde \gamma }^{s - 1}}} \right\} = \prod\nolimits_{k = 1}^K {\mathbb E\left\{ {{{\left( {1 + \frac{{{\alpha _{2,k}}{P_2}}}{{{\alpha _{2,k}}{P_1} + 1}}} \right)}^{s - 1}}} \right\}} \notag\\
 &= \prod\limits_{k = 1}^K {{\frac{1}{{{{\bar \alpha }_2}}}}\int_0^\infty  {{{\left( {1 + \frac{{{x_k}{P_2}}}{{{x_k}{P_1} + 1}}} \right)}^{s - 1}}{e^{ - \frac{{{x_k}}}{{{{\bar \alpha }_2}}}}}} d{x_k}}.
\end{align}
By making the change of variable ${z_k} = {c}/{\left({{x_k}{P_1} + 1}\right)}$, we have
\begin{align}\label{eqn:gamma_tilde_mellin_cov}
\mathbb E\left\{ {{{\tilde \gamma }^{s - 1}}} \right\} &= \prod\nolimits_{k = 1}^K {\frac{c}{{{P_1}{{\bar \alpha }_2}}}{e^{\frac{1}{{{P_1}{{\bar \alpha }_2}}}}}} \notag\\
&\times \int_0^c {{{\left( {1 + c - {z_k}} \right)}^{s - 1}}{e^{ - \frac{c}{{{P_1}{{\bar \alpha }_2}}}\frac{1}{{{z_k}}}}}\frac{1}{{{z_k}^2}}} d{z_k}.
\end{align}
By using \cite[eq.8.3.15]{debnath2010integral}, the CDF of $\tilde \gamma$ is obtained by using the inverse Mellin transform as
\begin{align}\label{eqn:gamm_tilde_cdf_mellin_inverse}
&{F_{\tilde \gamma }}\left( \gamma  \right) = \frac{-1}{{2\pi \rm i}}\int_{b - {\rm i}\infty }^{b + {\rm i}\infty } {\frac{{\mathbb E\left\{ {{{\tilde \gamma }^s}} \right\}}}{s}{\gamma ^{ - s}}ds}\notag\\
&= {\left( {\frac{c}{{{{\bar \alpha }_2}{P_1}}}} \right)^K}{e^{\frac{K}{{{{\bar \alpha }_2}{P_1}}}}}\int\nolimits_0^c { \cdots \int\nolimits_0^c {{e^{ - \frac{c}{{{{\bar \alpha }_2}{P_1}}}\sum\nolimits_{k = 1}^K {\frac{1}{{{z_k}}}} }}\prod\nolimits_{k = 1}^K {\frac{1}{{{z_k}^2}}} } }\notag\\
& \times \frac{- 1}{{2\pi \rm i}}\int_{b - {\rm i}\infty }^{b + {\rm i}\infty } { \frac{1}{s}\prod\nolimits_{k = 1}^K {{{\left( {1 + c - {z_k}} \right)}^s}} {\gamma ^{ - s}}ds} d{z_1} \cdots d{z_K},
\end{align}
where $b<0$. By identifying the contour integral in \eqref{eqn:gamm_tilde_cdf_mellin_inverse} with unit step function, ${F_{\tilde \gamma }}\left( \gamma  \right)$ can be derived as \eqref{eqn:cdf_gamma_sum_ir_fin}.

\section{Proof of \eqref{eqn:psi_lower_upper}}\label{app:psi_lb}
With the definition of ${\psi _K}\left( {{\gamma}} \right)$, ${\psi _K}\left( {{\gamma}} \right)$ can be rewritten as \eqref{eqn:psi_K}, shown at the top of the next page.
\begin{figure*}[!t]
\begin{align}\label{eqn:psi_K}
&{\psi _K}\left( {{\gamma}} \right) = \int\nolimits_0^c {\cdots\int\nolimits_0^c {  \int\nolimits_{\min \left( {\max \left( {1 + c - \frac{{{\gamma}}}{{\prod\nolimits_{k = 1}^{K - 1} {\left( {1 + c - {z_k}} \right)} }},0} \right),c} \right)}^c {{e^{ - \frac{c}{{\bar \alpha_2 {P_1}}}\sum\nolimits_{k = 1}^K {\frac{1}{{{z_k}}}} }}\prod\nolimits_{k = 1}^K {\frac{1}{{{z_k}^2}}} d{z_1}d{z_2} \cdots d{z_K}} } } \notag\\
&={\int _{\scriptstyle{z_{1:K - 1}} \in \left[ {0,c} \right],\hfill\atop
\scriptstyle\frac{\gamma }{{1 + c}} \le \prod\nolimits_{k = 1}^{K - 1} {\left( {1 + c - {z_k}} \right)}  \le \gamma \hfill}}{e^{ - \frac{c}{{{{\bar \alpha }_2}{P_1}}}\sum\nolimits_{k = 1}^{K - 1} {\frac{1}{{{z_k}}}} }}\prod\nolimits_{k = 1}^{K - 1} {\frac{1}{{{z_k}^2}}} d{z_1} \cdots d{z_{K - 1}}\int\nolimits_{1 + c - \frac{\gamma }{{\prod\nolimits_{k = 1}^{K - 1} {\left( {1 + c - {z_k}} \right)} }}}^c {{e^{ - \frac{c}{{{{\bar \alpha }_2}{P_1}}}\frac{1}{{{z_K}}}}}\frac{1}{{{z_K}^2}}d{z_K}}  \notag\\
&+ {\int _{\scriptstyle{z_{1:K - 1}} \in \left[ {0,c} \right], \hfill\atop
\scriptstyle\prod\nolimits_{k = 1}^{K - 1} {\left( {1 + c - {z_k}} \right)}  \le \frac{\gamma }{{1 + c}}}}{e^{ - \frac{c}{{{{\bar \alpha }_2}{P_1}}}\sum\nolimits_{k = 1}^{K - 1} {\frac{1}{{{z_k}}}} }}\prod\nolimits_{k = 1}^{K - 1} {\frac{1}{{{z_k}^2}}} d{z_1} \cdots d{z_{K - 1}}\int\nolimits_0^c {{e^{ - \frac{c}{{{{\bar \alpha }_2}{P_1}}}\frac{1}{{{z_K}}}}}\frac{1}{{{z_K}^2}}d{z_K}}.
\end{align}
\end{figure*}
\eqref{eqn:psi_K} can be further simplified by using the definition of ${\psi _K}\left( {{\gamma}} \right)$ as \eqref{eqn:psi_K_de}, shown at the top of the next page.
\begin{figure*}[!t]
\begin{align}\label{eqn:psi_K_de}
{\psi _K}\left( {{\gamma}} \right) &=\frac{{{{\bar \alpha }_2}{P_1}}}{c}
{\int _{\scriptstyle{z_{1:K - 1}} \in \left[ {0,c} \right],\hfill\atop
\scriptstyle\frac{\gamma }{{1 + c}} \le \prod\nolimits_{k = 1}^{K - 1} {\left( {1 + c - {z_k}} \right)}  < \gamma \hfill}}{e^{ - \frac{c}{{{{\bar \alpha }_2}{P_1}}}\sum\nolimits_{k = 1}^{K - 1} {\frac{1}{{{z_k}}}} }}\prod\nolimits_{k = 1}^{K - 1} {\frac{1}{{{z_k}^2}}} \left( {{e^{ - \frac{1}{{{{\bar \alpha }_2}{P_1}}}}} - {e^{ - \frac{c}{{{{\bar \alpha }_2}{P_1}}}\frac{1}{{1 + c - \frac{\gamma }{{\prod\nolimits_{k = 1}^{K - 1} {\left( {1 + c - {z_k}} \right)} }}}}}}} \right)d{z_1} \cdots d{z_{K - 1}}
  \notag\\
 &+ \frac{{{{\bar \alpha }_2}{P_1}}}{c}{e^{ - \frac{1}{{{{\bar \alpha }_2}{P_1}}}}}{\psi _{K - 1}}\left( {\frac{\gamma }{{1 + c}}} \right).
\end{align}
\end{figure*}
Since the first integral term is non-negative, the lower bound of ${\psi _K}\left( {{\gamma}} \right)$ in \eqref{eqn:psi_lower_upper} can be substantiated.
To obtain the upper bound of ${\psi _K}\left( \gamma  \right)$, \eqref{eqn:psi_K_de} can be rewritten as \eqref{eqn:psi_upper_first}, shown at the top of the next page, where $0 < \Delta \le c$.
\begin{figure*}[!t]
\begin{align}\label{eqn:psi_upper_first}
{\psi _K}\left( \gamma  \right) &= \frac{{{{\bar \alpha }_2}{P_1}}}{c}{\int _{\scriptstyle{z_{1:K - 1}} \in \left[ {0,c} \right],\hfill\atop
\scriptstyle\frac{\gamma }{{1 + c - \Delta }} \le \prod\nolimits_{k = 1}^{K - 1} {\left( {1 + c - {z_k}} \right)}  < \gamma \hfill}}{e^{ - \frac{c}{{{{\bar \alpha }_2}{P_1}}}\sum\nolimits_{k = 1}^{K - 1} {\frac{1}{{{z_k}}}} }}\prod\nolimits_{k = 1}^{K - 1} {\frac{1}{{{z_k}^2}}} \left( {{e^{ - \frac{1}{{{{\bar \alpha }_2}{P_1}}}}} - {e^{ - \frac{c}{{{{\bar \alpha }_2}{P_1}}}\frac{1}{{1 + c - \frac{\gamma }{{\prod\nolimits_{k = 1}^{K - 1} {\left( {1 + c - {z_k}} \right)} }}}}}}} \right)d{z_1} \cdots d{z_{K - 1}}\notag\\
& + \frac{{{{\bar \alpha }_2}{P_1}}}{c}{\int _{\scriptstyle{z_{1:K - 1}} \in \left[ {0,c} \right],\hfill\atop
\scriptstyle\frac{\gamma }{{1 + c}} \le \prod\nolimits_{k = 1}^{K - 1} {\left( {1 + c - {z_k}} \right)}  < \frac{\gamma }{{1 + c - \Delta }}\hfill}}{e^{ - \frac{c}{{{{\bar \alpha }_2}{P_1}}}\sum\nolimits_{k = 1}^{K - 1} {\frac{1}{{{z_k}}}} }}\prod\nolimits_{k = 1}^{K - 1} {\frac{1}{{{z_k}^2}}} \left( {{e^{ - \frac{1}{{{{\bar \alpha }_2}{P_1}}}}} - {e^{ - \frac{c}{{{{\bar \alpha }_2}{P_1}}}\frac{1}{{1 + c - \frac{\gamma }{{\prod\nolimits_{k = 1}^{K - 1} {\left( {1 + c - {z_k}} \right)} }}}}}}} \right)d{z_1} \cdots d{z_{K - 1}}\notag\\
& + \frac{{{{\bar \alpha }_2}{P_1}}}{c}{e^{ - \frac{1}{{{{\bar \alpha }_2}{P_1}}}}}{\psi _{K - 1}}\left( {\frac{\gamma }{{1 + c}}} \right).
\end{align}
\end{figure*}
Similarly to \eqref{eqn:g_red_approx1}, ${\psi _K}\left( \gamma  \right)$ is upper bounded as \eqref{eqn:psi_upper}, shown at the top of the next page.
\begin{figure*}[!t]
\begin{align}\label{eqn:psi_upper}
&{\psi _K}\left( \gamma  \right) \le \frac{{{{\bar \alpha }_2}{P_1}}}{c}\left( {{e^{ - \frac{1}{{{{\bar \alpha }_2}{P_1}}}}} - {e^{ - \frac{1}{{{{\bar \alpha }_2}{P_1}}}\frac{c}{\Delta }}}} \right)\underbrace {{{\int _{\scriptstyle{z_{1:K - 1}} \in \left[ {0,c} \right],\hfill\atop
\scriptstyle\frac{\gamma }{{1 + c - \Delta }} \le \prod\nolimits_{k = 1}^{K - 1} {\left( {1 + c - {z_k}} \right)}  < \gamma \hfill}}{e^{ - \frac{c}{{{{\bar \alpha }_2}{P_1}}}\sum\nolimits_{k = 1}^{K - 1} {\frac{1}{{{z_k}}}} }}\prod\nolimits_{k = 1}^{K - 1} {\frac{1}{{{z_k}^2}}} d{z_1} \cdots d{z_{K - 1}}}}_{ \le {\psi _K}\left( \gamma  \right)}\notag\\
& + \frac{{{{\bar \alpha }_2}{P_1}}}{c}{e^{ - \frac{1}{{{{\bar \alpha }_2}{P_1}}}}}\underbrace {{{\int _{\scriptstyle{z_{1:K - 1}} \in \left[ {0,c} \right],\hfill\atop
\scriptstyle\frac{\gamma }{{1 + c}} \le \prod\nolimits_{k = 1}^{K - 1} {\left( {1 + c - {z_k}} \right)}  < \frac{\gamma }{{1 + c - \Delta }}\hfill}}{e^{ - \frac{c}{{{{\bar \alpha }_2}{P_1}}}\sum\nolimits_{k = 1}^{K - 1} {\frac{1}{{{z_k}}}} }}\prod\nolimits_{k = 1}^{K - 1} {\frac{1}{{{z_k}^2}}} d{z_1} \cdots d{z_{K - 1}}}}_{ = {\psi _{K - 1}}\left( {\frac{\gamma }{{1 + c - \Delta }}} \right) - {\psi _K}\left( {\frac{\gamma }{{1 + c}}} \right)} + \frac{{{{\bar \alpha }_2}{P_1}}}{c}{e^{ - \frac{1}{{{{\bar \alpha }_2}{P_1}}}}}{\psi _{K - 1}}\left( {\frac{\gamma }{{1 + c}}} \right).
\end{align}
\hrulefill
\end{figure*}
By applying $1 - {e^{ - x}} \le x$ to \eqref{eqn:psi_upper}, ${\phi _K}\left( \gamma  \right)$ is thus upper bounded as \eqref{eqn:psi_lower_upper}.

\section{Proof of Lemma \ref{the:psi_asy}}\label{app:proof_psi_gener}
\subsection{Case 1: $\gamma=1$}
If $\gamma =1$, ${\psi _K}\left( \gamma \right)$ is founded to be zero from its definition in \eqref{eqn:cdf_gamma_sum_ir_fin} because $u\left( {\gamma  - \prod\nolimits_{k = 1}^K {\left( {1 + c - {z_k}} \right)} } \right)=0$, where $z_k \in [0,c]$.
\subsection{Case 2: $1<\gamma<1+c$}
For the case of $1<\gamma<1+c$, the set of $\{\left( {{z_1},\cdots,{z_K}} \right)\}$ constituted by $\prod\nolimits_{k = 1}^K {\left( {1 + c - {z_k}} \right)}  \le {\gamma}$ is a subset of $\left\{ {\left( {{z_1},\cdots,{z_K}} \right):{z_k} \ge 1+c- \gamma} \right\}$ because $0 \le z_k \le c$. Accordingly, ${\psi _K}\left( {{\gamma}}  \right)$ is upper bounded as
\begin{align}\label{eqn:phi_K_lower}
&{\psi _K}\left( \gamma \right)\le \int\nolimits_{1+c- \gamma}^c { \cdots \int\nolimits_{1+c- \gamma }^c {{e^{ - \frac{c}{{\bar \alpha_2 {P_1}}}\sum\nolimits_{k = 1}^K {\frac{1}{{{z_k}}}} }}\prod\nolimits_{k = 1}^K {\frac{1}{{{z_k}^2}}} d{z_1} \cdots d{z_K}} } \notag\\
 &= {\left( {\frac{{\bar \alpha_2 {P_1}}}{c}\left( {{e^{ - \frac{1}{{\bar \alpha_2 {P_1}}}}} - {e^{ - \frac{c}{{\bar \alpha_2 {P_1}\left({1+c - \gamma }\right)}}}}} \right)} \right)^K}\notag \\
 &=e^{ - K\varepsilon_3 }{\left( {\frac{1}{{1+c - \gamma }} - \frac{1}{c}} \right)^K},
\end{align}
where the last second step holds by using Lagrange's mean value theorem and ${1}/{\left({\bar \alpha_2 {P_1}}\right)} \le \varepsilon_3  \le {c}/{\left({\bar \alpha_2 {P_1}}{(1+c-\gamma) }\right)}$. Since $\varepsilon_3 \ge {1}/{\left({\bar \alpha_2 {P_1}}\right)} $, the upper bound of \eqref{eqn:phi_K_lowerupper} is thus demonstrated.

Moreover, $\left\{ {\left( {{z_1},\cdots,{z_K}} \right):{z_k} \ge 1 + c - \sqrt[K]{\gamma}} \right\}$ is readily found to be a subset of $\prod\nolimits_{k = 1}^K {\left( {1 + c - {z_k}} \right)}  \le {\gamma}$, we reach
\begin{align}\label{eqn:phi_K_upper}
&{\psi _K}\left( \gamma  \right) \notag\\
&\ge \int\nolimits_{1+c-\sqrt[K]{\gamma}}^c { \cdots \int\nolimits_{1+c-\sqrt[K]{\gamma}}^c {{e^{ - \frac{c}{{\bar \alpha_2 {P_1}}}\sum\nolimits_{k = 1}^K {\frac{1}{{{z_k}}}} }}\prod\nolimits_{k = 1}^K {\frac{1}{{{z_k}^2}}} d{z_1} \cdots d{z_K}} } \notag\\
 &= {\left( {\frac{{\bar \alpha_2 {P_1}}}{c}\left( {{e^{ - \frac{1}{{\bar \alpha_2 {P_1}}}}} - {e^{ - \frac{c}{{\bar \alpha_2 {P_1}}{\left(1+c-\sqrt[K]{\gamma}\right)}}}}} \right)} \right)^K}\notag\\
 & = {e^{ - K{\varepsilon _4}}}{\left( {\frac{1}{{1 + c - \sqrt[K]{\gamma }}} - \frac{1}{c}} \right)^K},
\end{align}
where the last step holds in analogous to \eqref{eqn:phi_K_lower} and ${1}/{\left({\bar \alpha_2 {P_1}}\right)} \le \varepsilon_4  \le {c}/{\left({\bar \alpha_2 {P_1}}{(1+c-\sqrt[K]{\gamma}) }\right)}$. Considering $\varepsilon_4  \le {c}/{\left({\bar \alpha_2 {P_1}}{(1+c-\sqrt[K]{\gamma}) }\right)}$, ${\psi _K}\left( \gamma  \right) $ is lower bounded as \eqref{eqn:phi_K_lowerupper}. Since ${\psi _K}\left( \gamma  \right) $ is an increasing function of $P_1$ and it is also upper bounded, the limit of $ {\psi _K}\left( \gamma  \right)$ exists as $P_1 \to \infty$. With \eqref{eqn:phi_K_lowerupper},
$\mathop {\lim }\nolimits_{{P_1} \to \infty } {\psi _K}\left( \gamma  \right)$ is bounded as
\begin{align}\label{eqn:phi_K_boudns}
{\left( {\frac{1}{{1+c-\sqrt[K]{\gamma}}} - \frac{1}{c}} \right)^K} &\le \mathop {\lim }\limits_{{P_1} \to \infty } {\psi _K}\left( \gamma  \right)\notag\\
&\le {\left( {\frac{1}{{1+c-\gamma}} - \frac{1}{c}} \right)^K}.
\end{align}
Hence, the limit of $ {\psi _K}\left( \gamma  \right)$ as $P_1 \to \infty$ is a non-zero constant if $1<\gamma<1+c$.
\subsection{Case 3: $\gamma\ge(1+c)^{K}$}
If $\gamma\ge(1+c)^{K}$, $u\left( {\gamma  - \prod\nolimits_{k = 1}^K {\left( {1 + c - {z_k}} \right)} } \right)=1$ follows. In analogous to  \eqref{eqn:phi_K_lower_ccc3}, $ {\psi _K}\left( \gamma  \right)$ reduces to
\begin{align}\label{eqn:psi_K_lower_ccc3}
{\psi _K}\left( \gamma  \right) 
 &={{\left( {\frac{{\bar \alpha_2 }{P_1}}{c}} \right)}^K}{e^{ - \frac{K}{{\bar \alpha_2 {P_1}}}}}=\Theta({P_1}^K).
\end{align}
\subsection{Case 4: $(1+c)^k\le\gamma<(1+c)^{k+1}$ for $k \in [1,K-1]$}
Similarly to Appendix \ref{sec:case4_cc}, the case of $(1+c)^k<\gamma<(1+c)^{k+1}$ is first considered and that of $\gamma = (1+c)^k$ is specially treated. By repeatedly using \eqref{eqn:psi_lower_upper} together with the results in cases 2 and 3, the lower and upper bounds of ${\psi _K}\left( \gamma  \right)$ show ${\psi _K}\left( \gamma  \right)=\Theta({P_1}^k)$. Similarly, it is noteworthy that $\Delta$ should be properly chosen to obtain a tight upper bound for ${\psi _K}\left( \gamma  \right) $. Specifically, $\Delta$ could be assigned with $\Delta = (1+c-\exp({\rm mod}(\ln\gamma,\ln(1+c))))/2$ in each recursive step.

On the other hand, if $\gamma = (1+c)^k$, by successively applying the upper bound in \eqref{eqn:psi_lower_upper}, we have ${\phi _K}\left( (1+c)^k  \right) \le \Theta({P_1}^k)$. By repeatedly using the lower bound of ${\phi _\kappa}\left( c  \right)$, it follows that
\begin{equation}\label{eqn:psi_lowerck}
{\psi _K}\left( (1+c)^k   \right) \ge  \left(\frac{{\bar \alpha_2 {P_1}}}{c}\right)^{k-1}{e^{ - \frac{k-1}{{\bar \alpha_2 {P_1}}}}}{\psi _{\kappa}}\left( {{{1 + c}}} \right),
\end{equation}
where $\kappa=K-k+1$. With the definition of ${\psi _K}\left( \gamma   \right)$ in \eqref{eqn:cdf_gamma_sum_ir_fin}, ${{\psi _\kappa}\left( 1+c \right)}$ can be rewritten similarly to \eqref{eqn:phi_K_c} as
\begin{align}\label{eqn:psi_K_c}
&{\psi _\kappa }\left( {1 + c} \right) = \int_0^c {{e^{ - \frac{c}{{{{\bar \alpha }_2}{P_1}}}\frac{1}{{{z_\kappa }}}}}\frac{1}{{{z_\kappa }^2}}d{z_\kappa }} \int_0^c { \cdots \int_0^c {{e^{ - \frac{c}{{{{\bar \alpha }_2}{P_1}}}\sum\nolimits_{k = 1}^\kappa  {\frac{1}{{{z_k}}}} }}} } \notag\\
&\times \prod\limits_{k = 1}^\kappa  {\frac{1}{{{z_k}^2}}} u\left( {\frac{{1 + c}}{{1 + c - {z_\kappa }}} - \prod\nolimits_{k = 1}^{\kappa  - 1} {\left( {1 + c - {z_k}} \right)} } \right)d{z_1} \cdots d{z_{\kappa  - 1}}\notag\\
& = \int_0^c {{e^{ - \frac{c}{{{{\bar \alpha }_2}{P_1}}}\frac{1}{{{z_\kappa }}}}}\frac{1}{{{z_\kappa }^2}}{\psi _{\kappa  - 1}}\left( {\frac{{1 + c}}{{1 + c - {z_\kappa }}}} \right)d{z_\kappa }}.
\end{align}
Thanks to the increasing monotonicity of ${\psi _\kappa}\left( z  \right)$ with respect to $z$, \eqref{eqn:psi_K_c} can be further expressed by using the intermediate value theorem as
\begin{align}\label{eqn:psi_K_c_fin}
&{\psi _\kappa }\left( 1+c \right) = {\psi _{\kappa - 1}}\left( {\frac{{1 + c}}{{1 + c - \varsigma _{P_1}}}}  \right)\int_0^c {{{e^{ - \frac{c}{{{P_1}{{\bar \alpha }_2}}}\frac{1}{{{z_\kappa}}}}}\frac{1}{{{z_\kappa}^2}}}d{z_\kappa}} \notag\\
&= \frac{{P_1}{{\bar \alpha }_2}}{c}{e^{ - \frac{1}{{{P_1}{{\bar \alpha }_2}}}}}{\psi _{\kappa - 1}}\left( {\frac{{1 + c}}{{1 + c - \varsigma _{P_1}}}}  \right),
\end{align}
where $\varsigma_{P_1}  \in \left( {0,c} \right)$. Proceeding as in the proof of Theorem \ref{the:inter_med}, with the increase of $P_1$ and the condition $1 < {\left({{1 + c}}\right)/\left({{1 + c - \varsigma _{P_1}}}\right)} < 1+c$, $\mathop {\lim }\nolimits_{{P_1} \to \infty } {\psi _{\kappa - 1}}\left( {{\left({1 + c}\right)}/{\left({1 + c - \varsigma _{P_1}}\right)}}  \right) = {\rm const} \ne 0$ and $\mathop {\lim }\nolimits_{{P_1} \to \infty } \varsigma_{P_1} \ne 0$ can be proved by contradiction. Therefore, by applying the result of Case 2 to \eqref{eqn:psi_K_c_fin}, we arrive at ${\psi _\kappa }\left( {1 + c} \right) =\Theta(P_1)$. By combining this result and \eqref{eqn:psi_lowerck}, ${\psi _K}\left( (1+c)^k  \right) \ge \Theta({P_1}^k)$ follows. Together with the upper bounds of ${\psi _K}\left( (1+c)^k  \right)$, applying the squeeze theorem leads to ${\psi _K}\left( (1+c)^k   \right) = \Theta({P_1}^k)$. We thus accomplish the proof.

\section{Proof of \eqref{eqn:asy3}}\label{app:div_eff}
Based on whether the accumulated mutual information in \eqref{eqn:noma_harq_inf_i2_eff} is a maximization or summation of multiple random variables, the proof of \eqref{eqn:asy3} is divided into the following two cases.
\subsection{Type I HARQ-aided NOMA system}
By using \eqref{eqn:noma_harq_inf_i2_eff} along with the independence between channel gains, we have
\begin{align}\label{eqn:out2_I_eff}
&\Pr \left\{ {{I_{2 \to 2,K,\ell }} < {R_2}} \right\}\notag \\
 &=\Pr \left\{ {\max \left\{ {{{\log }_2}\left( {1 + \frac{{{\alpha _{2,k}}{P_2}}}{{{\alpha _{2,k}}{P_1} + 1}}} \right):k \in \left[ {1,\ell } \right]} \right\} < {R_2}} \right\}\notag\\
 &\times \Pr \left\{ {\max \left\{ {{{\log }_2}\left( {1 + {\alpha _{2,k}}{P_2}} \right):k \in \left[ {\ell  + 1,K} \right]} \right\} < {R_2}} \right\}\notag\\
 &= \Theta \left( {{P_2}^{ - \ell {{\left[ {1 - \left\lfloor {\frac{{{2^{{R_2}}} - 1}}{c}} \right\rfloor } \right]}^ + }}} \right)\Theta \left( {{P_2}^{ - \left( {K - \ell } \right)}} \right),
\end{align}
where the last step holds by using \eqref{eqn:diver_order_I} and \cite{shi2016optimal}. Applying the property of the product of big-Theta notation \eqref{eqn:out2_I_eff} finally gives rise to \eqref{eqn:asy3}.
\subsection{HARQ-CC and -IR-aided NOMA systems}
We first derive the asymptotic outage behavior for the HARQ-CC-aided NOMA system. For simplicity, we define $Y = {\sum\nolimits_{k = 1}^\ell  {{{{\alpha _{2,k}}{P_2}}}/{({{\alpha _{2,k}}{P_1} + 1})}} }$, $Z={\sum\nolimits_{k = \ell  + 1}^K {{\alpha _{2,k}}{P_2}} }$, $\Pr \left\{ {{I_{2 \to 2,K,\ell }} < {R_2}} \right\}$ can be rewritten by using \eqref{eqn:noma_harq_inf_i2_eff} as
\begin{align}\label{eqn:out2_cc_eff}
 &\Pr \left\{ {{I_{2 \to 2,K,\ell }} < {R_2}} \right\}= \Pr \left\{ {Y + Z < {2^{{R_2}}} - 1} \right\}\notag\\
& = \int_0^{{2^{{R_2}}} - 1} {\Pr \left\{ {Y < {2^{{R_2}}} - 1 - z} \right\}{f_Z}\left( z \right)dz}.
\end{align}
By using \eqref{eqn:diversity_order_cc}, \eqref{eqn:out2_cc_eff} can be simplified as
\begin{align}\label{eqn:out2_cc_effsim}
 &\Pr \left\{ {{I_{2 \to 2,K,\ell }} < {R_2}} \right\}\notag\\
 &= \int_0^{{2^{{R_2}}} - 1} {\Theta \left( {{P_2}^{ - {{\left[ {\ell  - \left\lfloor {\frac{{{2^{{R_2}}} - 1 - z}}{c}} \right\rfloor } \right]}^ + }}} \right){f_Z}\left( z \right)dz} \notag\\
&= \sum\nolimits_{k = 0}^{\left\lfloor {\frac{{{2^{{R_2}}} - 1}}{c}} \right\rfloor } {\Theta \left( {{P_2}^{ - {{\left[ {\ell  - \left\lfloor {\frac{{{2^{{R_2}}} - 1}}{c}} \right\rfloor  + k} \right]}^ + }}} \right)}\notag\\
&\quad \times \left( {{F_Z}\left( {{z_{k + 1}}} \right) - {F_Z}\left( {{z_k}} \right)} \right).
\end{align}
where we stipulate ${z_0}=0$ and ${z_k} = {2^{{R_2}}} - 1 + c\left( {k - \left\lfloor {{({{2^{{R_2}}} - 1})}/{c}} \right\rfloor } \right)$ for $k\ge 1$. It has been substantiated in \cite{shi2016optimal} that ${{F_Z}\left( {{z_{k}}} \right)}={\Theta \left( {{P_2}^{ - \left( {K - \ell } \right)}} \right)}$. Thus, it follows that
\begin{align}\label{eqn:out2_cc_effsimfin}
 &\Pr \left\{ {{I_{2 \to 2,K,\ell }} < {R_2}} \right\} \notag\\
 &= \sum\nolimits_{k = 0}^{\left\lfloor {\frac{{{2^{{R_2}}} - 1}}{c}} \right\rfloor } {\Theta \left( {{P_2}^{ - {{\left[ {\ell  - \left\lfloor {\frac{{{2^{{R_2}}} - 1}}{c}} \right\rfloor  + k} \right]}^ + } - \left( {K - \ell } \right)}} \right)} \notag\\
 &= \Theta \left( {{P_2}^{ - \min \left\{ {{{\left[ {\ell  - \left\lfloor {\frac{{{2^{{R_2}}} - 1}}{c}} \right\rfloor  + k} \right]}^ + } + K - \ell :k \in \left[ {0,\left\lfloor {\frac{{{2^{{R_2}}} - 1}}{c}} \right\rfloor } \right]} \right\}}} \right).
\end{align}
\eqref{eqn:asy3} thus follows for the HARQ-CC-aided NOMA system. Likewise, by using the results in \cite{shi2017asymptotic}, we can obtain the asymptotic scaling law of $\Pr \left\{ {{I_{2 \to 2,K,\ell }} < {R_2}} \right\}$ for the HARQ-IR-aided NOMA system as \eqref{eqn:asy3}.

\bibliographystyle{ieeetran}
\bibliography{manuscript_1}

\end{document}